\documentclass[runningheads]{llncs}
\usepackage{lineno}
\usepackage[utf8]{inputenc}
\usepackage[T1]{fontenc}
\usepackage{tikz}
\usepackage{amsmath}
\usepackage{algorithm2e}
\usepackage{cite}           
\usepackage{lipsum}
\usepackage{tikz}
\usepackage{calc}
\usepackage{cite}           
\usepackage{enumitem}
\usepackage{float}
\usepackage{amsmath, amssymb}
\usepackage{amsfonts}
\usepackage{wrapfig}
\usepackage{multicol}
\usepackage{listings}
\usepackage{multirow}
\usepackage{varwidth}
\usepackage[section]{placeins}
\usepackage{hyperref}
\usepackage{cleveref}
\usepackage{subcaption}
\usepackage{tablefootnote}
\usepackage{todonotes}
\usepackage{thmtools}
\usepackage{thm-restate}

\usepackage{graphicx}
\usepackage[noend]{algpseudocode}

\newcommand{\pre}[1]{{^{\bullet}#1}}
\newcommand{\post}[1]{{{#1}^{\bullet}}}

\usepackage{orcidlink}

\usepackage{mathtools}

\algnewcommand\LeftComment[2]{%
\hspace{#1\algindent}$\triangleright$ \eqparbox{}{#2} \hfill %
}
\title{ \bf Constructing Weakly Terminating \\ Interface Protocols}
\author{
Debjyoti Bera\inst{1}\orcidlink{0009-0006-5362-8995} \and Tim~A.C.~Willemse\inst{2}\orcidlink{0000-0003-3049-7962}
}
\institute{
TNO-ESI, Eindhoven, The Netherlands\\
\email{debjyoti.bera@tno.nl},\\ 
\and
TU Eindhoven, The Netherlands\\
\email{t.a.c.willemse@tue.nl}
}

\begin{document}

\maketitle
\thispagestyle{empty}
\pagestyle{plain}

\begin{abstract}
Interfaces play a central role in determining compatible component compositions by prescribing permissible interactions between a service provider (server) and its consumers (clients). The high degree of concurrency in asynchronous communicating systems increases the risk of unintentionally introducing deadlocks and livelocks. The weak termination property serves as a basic sanity check to avoid such problems. It assures that in each reachable state, the system has the option to eventually terminate. This paper generalizes existing results that, by construction, guarantee weakly terminating interface compositions. Our generalizations make the theory applicable more broadly in practice. Starting with an interface specification of a server satisfying certain properties, we show how a class of clients modeling different usage contexts can be derived using a partial mirroring relation. 
Furthermore, we discuss an embedding of our results in an open-source tool to guide modelers in designing weakly terminating interfaces. 
\end{abstract}

\section{Introduction}
The first step in the design of a component-based system involves the decomposition of a given system design problem into functionally coherent units of execution called components \cite{mcilroy1968mass, szyperski2002component}. 
Starting from an abstract model, each component is further detailed by means of stepwise refinements. 
To accomplish a common goal, components may need to interact with each other. In modern distributed systems, these interactions are typically asynchronous via message passing. 
The set of permissible interactions between pairs of components defines an interface protocol (contract) having two sides: a server side providing the functionality and a client side consuming it. 
The use of state machine formalisms \cite{de2001interface} to specify such interactions is quite prevalent.  
When interfaces between components have been defined, the resulting system is obtained by composing shared client-server interfaces satisfying some pre-defined notion of compatibility. 

A direct consequence of a component-based approach is the high degree of complexity induced by concurrency and asynchronous communication making it difficult to avoid subtle yet critical design errors, resulting in deadlocks or livelocks. 
It turns out that even if a server side protocol is free of these issues, simply mirroring this protocol to obtain client side protocols to define their usage does not guarantee the absence of errors.
Therefore the use of formal modeling and analysis techniques becomes inevitable. Petri nets provide a natural way to model an asynchronous communicating system and are supported by a rich set of verification methods and tools \cite{amat2023behind}. The graphical representation of a Petri net structure provides an intuitive way of modeling systems and in particular, structural analysis techniques overcome the drawbacks of state space exploration techniques for verifying system behavior.

A basic requirement for any component-based system is that all components in the system should always have the option to reach a desired, final state, and that all messages that have been sent have been received. This property is called \emph{weak termination}. Weak termination does not imply that the system always terminates but that there is always an option to do so. As a consequence of this property the freedom of deadlocks and livelocks is guaranteed.

Many interesting results exist for subclasses of Petri nets such that weak termination is implied by their structure. In \cite{van2009compositional} a sufficient condition is presented to pairwise verify weak termination for tree-structured compositions.
For a subclass of compositions of pairs of components, called ATIS-nets, this condition
is implied by their structure \cite{van2010construction}. 
ATIS-nets are constructed from pairs of acyclic marked graphs and isomorphic state machines, and the simultaneous refinement of pairs of places with weakly terminating multi-workflow nets. 
In \cite{bera2011component}, an architectural framework based on the notion of components and ports (defining a communication protocol) is proposed. The sufficient conditions for port compatibility, as well as on the acyclic network structure of component compositions, induces the weak termination property on the resulting system (network of asynchronous communicating components).
In \cite{bera2012designing}, another architectural framework based on the notion of components and interaction patterns (a parameterized multi-workflow net defining a communication protocol) is proposed. 
The interaction patterns cover three commonly occurring ones in practice, namely remote procedure call, cancellation (goal-feedback-result) and publish-subscribe. 
In each communication refinement of the existing system, a new component is added to the system as part of an interaction pattern. 
In \cite{BeraThesis}, the architectural framework \cite{bera2012designing} is extended with the mirrored port pattern by lifting the notion of compatible ports from \cite{bera2011component} to interaction patterns.

\begin{figure}[t] 
    \centering
    \includegraphics[scale=0.57]{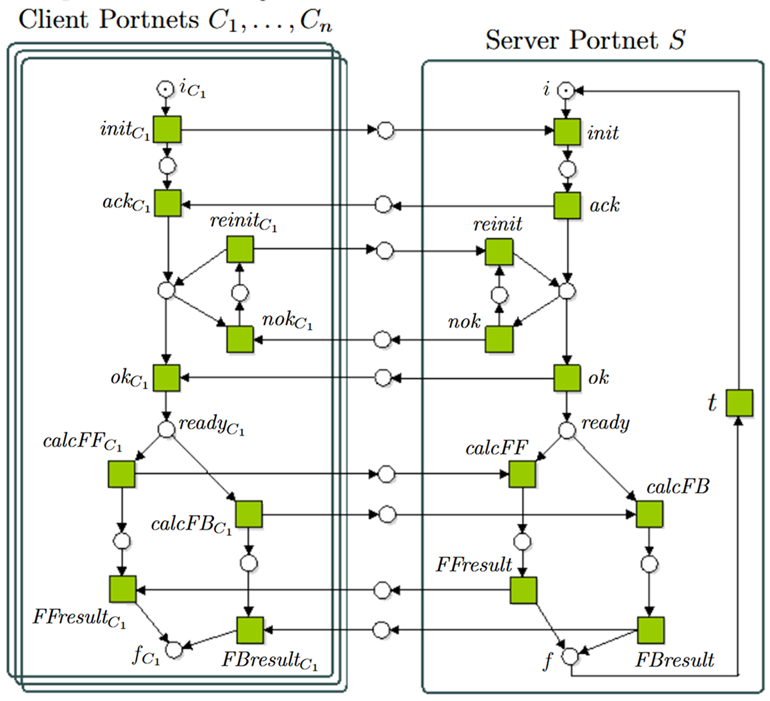}
    \caption{The Mirrored Portnet Interaction Pattern}
    \label{fig:mirroredPortnets}
\end{figure}

The mirrored port pattern defines a class of weakly terminating single server and multiple client interaction patterns, where each client is an isomorphic mirror of the server state machine satisfying structural properties that enforce communication conditions such as the \emph{choice} and \emph{leg} property. The former guarantees that only one of the parties make a choice whenever there is a conflict, while the latter guarantees that every path between a split and a join always requires the participation of a server and one of the clients. An example of a mirrored port pattern is shown in Fig.~\ref{fig:mirroredPortnets}.
It turns out that for practical applications, the mirrored port pattern is quite limiting for at least three reasons:
\begin{itemize}
\item[(i)] the choice property requires there are no races between a server and client, i.e., the initiative of resolving conflicting transitions is either with the server or one of its client, but not both. 
This restriction is quite limiting for practical applications where races are often intensional. For instance, in Fig.~\ref{fig:mirroredPortnets}, it is not possible to model a race between a client firing one of its transitions from the place $\textit{ready}$ and a server that may emit an \emph{error} from $\textit{ready}$. 
\item[(ii)] mirrored portnets do not allow multiple transitions in a server/client to be connected to the same interface place. This means that it is not possible to send/receive the same message from more than one state. For instance, in Fig.~\ref{fig:mirroredPortnets}, the transition $\textit{reinit}$ had to be introduced since, transition \textit{init} cannot occur twice and share the same interface place. 
\item[(iii)] the isomorphism requirement implies that client portnets must fully mirror the server portnet even though they may not use all functionality. For instance, in Fig.~\ref{fig:mirroredPortnets}, all clients $C_1,\ldots, C_n$ are required to implement both choices from the place $\textit{ready}$, even though some of the clients may only use one of them. 
This is a common requirement since each client may require a different usage of a server depending on the context of the component they occur.  
\end{itemize}

This paper revisits the mirrored port interaction pattern \cite{BeraThesis} and extends it with a labeling function, relaxes constraints on its structure and weakens the isomorphism requirement on portnets. We call this class of interaction patterns \emph{partial mirrored portnet} and show that the weak termination property holds for these patterns. 
The results provide a structured method to design and reason about the correctness of a large class of interfaces encountered in practice. Since partial mirrored portnets are an interaction pattern of the component-based framework, we may reuse the stepwise refinements based construction method to design a weakly terminating network of asynchronous communicating components starting from an architectural diagram. 
A first application of these results has been realized in the open-source ComMA framework \cite{kurtev2024model}, widely used by engineers in industry to specify software interfaces. Algorithms have been implemented to automatically detect violations of portnet properties in interface specifications and provide feedback about termination problems at design-time.  

This paper is structured as follows: Sec.~2 presents basic mathematical concepts and notations. Sec.~3 presents labeled portnets, compositions and the partial mirror relation. Sec.~4 presents restrictions for a labeled portnet to be well-formed. Sec.~5 proves weak termination in a composition of a well-formed labeled portnet and its partial mirror. Sec.~6 extends the results to multiple clients. Sec.~7 shows an application of these results in ComMA. The paper concludes in Sec.~8. 

\section{Preliminaries}
Let $S$ be a set.
A \emph{bag} over $S$ is a function $m: S \rightarrow \mathbb{N}$, where $\mathbb{N} = \{0,1,2,\ldots\}$ denotes the set of natural numbers. 
For $s \in S$, $m(s)$ denotes the number of occurrences of $s$ in $m$. 
The set of all bags over $S$ is denoted by $B(S)$. 
We use $+$ and $-$ to denote the standard addition and difference of two bags.
The relations $=, <,$ and $\leq$ compare bags element-wise.
We treat sets as bags in which each element of the set occurs exactly once in the bag. 

The set of finite sequences over $S$ is denoted by $S^*$, with $\epsilon$ denoting the empty sequence and $\langle a_0, \dots, a_n \rangle$ denoting a sequence of length $n+1$.
We use $\circ$ to concatenate two sequences and we denote the length of a sequence $\sigma$ by $|\sigma|$.
A prefix of a sequence $\sigma$ is a sequence $\sigma'$ such that $\sigma = \sigma' \circ \sigma''$ for some sequence $\sigma''$.
We write $x \in \sigma$ when there are sequences $\sigma', \sigma''$ such that $\sigma = \sigma' \circ \langle x \rangle \circ \sigma''$.
The projection of a sequence $\sigma$ onto a set $Q$, denoted $\sigma_{|Q}$, is defined inductively as follows: $\epsilon_{|Q} = \epsilon$ and $(\langle a \rangle \circ \sigma')_{|Q} = \langle a \rangle \circ (\sigma'_{|Q})$ if $a \in Q$, and $\sigma'_{|Q|}$ otherwise.
We denote the set obtained by interleaving sequences $\sigma$ and $\sigma'$ by $\sigma || \sigma'$.

\newcommand{\mL}{\mathcal{L}}
A \emph{Petri net with labels}, or simply Petri net from hereon, is a mathematical model for describing the dynamics of complex systems.
\begin{definition}
A \emph{Petri net} is a tuple $N = (P, T, F, \mL, \mu)$, where $P$ is the set of \emph{places}; $T$ is a set of \emph{transitions} such that $P \cap T = \emptyset$, $F$ is the \emph{flow relation} $F \subseteq (P \times T) \cup (T \times P)$, $\mL$ is a set of \emph{labels} and $\mu : T \to \mL$ is a labeling function. 
\end{definition}
Let $N = (P,T,F, \mL, \mu)$ be an arbitrary but fixed Petri net.
Elements from $P \cup T$ are called \emph{nodes} and elements from $F$ are referred to as \emph{arcs}. 
For clarity, we occasionally denote the places of net $N$ by $P_N$, transitions by $T_N$ and similarly for other elements of the tuple. 
We define the \emph{preset} of a node $x$ as $\pre{x} = \{x' | (x',x) \in F\}$ and the \emph{postset} as $\post{x} = \{x' | (x,x') \in F\}$ and lift these notations to sets of nodes in the standard way.

The set of paths $\textit{PS}(N)$ of $N$ is the smallest set satisfying $\epsilon \in \textit{PS}(N)$, $\langle x \rangle \in \textit{PS}(N)$ for all $x \in P \cup T$, and for all $(x_0,x_1) \in F$ and finite paths $\langle x_1, \dots, x_n \rangle \in \textit{PS}(N)$, also $\langle x_0, x_1, \dots, x_n \rangle \in \textit{PS}(N)$.
Petri net $N$ is \emph{strongly connected} if for all distinct $x,x' \in P \cup T$ there is a path $\langle x, \dots, x' \rangle \in \textit{PS}(N)$. 

\begin{definition}
We say that $N$ is a \emph{state machine} (S-net) iff for all $t \in T$: $|\pre{t}| \leq 1$ and $|\post{t}| \leq 1$. 
It is a \emph{workflow net} (WFN) if there is exactly one place $i$ with $\pre{i} = \emptyset$, called the \emph{initial place}, one place $f$ with $\post{f} = \emptyset$, called the \emph{final place}, and all nodes $n \in P \cup T$ are on a path from $i$ to $f$.
A place $p$ is called a \emph{split} if $|\post{p}| > 1$ and it is called a \emph{join} if $|\pre{p}| > 1$. 
\end{definition}

Petri nets $N$ and $M$ are disjoint if $(P_N \cup T_N) \cap (P_M \cup T_M) = \emptyset$. 
They are \emph{isomorphic}, denoted by $N \cong M$ iff there is a bijection $\psi : P_N \cup T_N \rightarrow P_M \cup T_M$ such that for all $p \in P_N$ we have $\psi(p) \in P_M$ and for all $t \in T_N$ we have $\psi(t) \in T_M$ and $\mu_N(t) = \mu_M(\psi(t))$, and for all $(x,y) \in (P_N \cup T_N) \times (P_N \cup T_N)$ we have $(x,y) \in F_N$ iff $(\psi(x), \psi(y)) \in F_M$. 
\smallskip

The \emph{state} of Petri net $N = (P, T, F,\mL,\mu)$ is determined by its \emph{marking} which represents the distribution of tokens over places of the net. 
Formally, a marking of $N$ is a bag $m \in B(P)$, where $m(p)$ denotes the number of tokens in place $p \in P$. If $m(p) > 0$, place $p \in P$ is called \emph{marked} in marking $m$. 
A transition $t \in T$ is \emph{enabled} in $m$ iff $\pre{t} \leq m$, i.e., all the places in its preset are marked. 
Transition $t$ may \emph{fire} from marking $m$ if it is enabled, resulting in a marking $m'$ that satisfies $m'(p) = m(p) - \pre{t}(p) + \post{t}(p)$ for all $p \in P$. 
We write $m \xrightarrow{t} m'$ if $t$ is enabled in $m$ and firing $t$ results in $m'$, and lift this notion to sequences of transitions in the expected manner.

A \emph{net system} is a Petri net with an initial and final marking; its semantics is given by a \emph{labeled transition system}.
\begin{definition}
Let $N = (P,T,F,\mL,\mu)$ be a Petri net.
A net system is a 3-tuple $S = (N, m_0, m_f)$, where $m_0 \in B(P)$ is the initial marking and $m_f \in B(P)$ is the final marking. 
\end{definition}
Given a net system $S$, a marking $m'$ of $S$ is \emph{reachable} from $m$ if for some $\sigma \in T^*$ we have $m \xrightarrow{\sigma} m'$; we denote the set of markings reachable from $m$ by $\mathcal{R}(N,m)$.
A place $p \in P_N$ in a net system is called \emph{safe} if for all reachable markings $m \in \mathcal{R}(N,m_0)$ we have $m(p) \leq 1$. 
An important correctness criterion of a net system is whether it is always able to terminate: it implies the absence of (unknown) deadlocks in a system, but also the ability to direct the behavior to a terminal state.
\begin{definition}
A net system is \emph{weakly terminating} if from all markings $m \in \mathcal{R}(N,m_0)$ the final marking $m_f$ is reachable.
\end{definition}

\section{Labeled Portnets and their Mirrors}
Open Petri nets (OPN) are a well-studied \cite{milner2003bigraphs, sassone2005congruence, baldan2005compositional,
baldan2015modular} class of Petri nets equipped with open places and visible transitions, i.e., distinguished sets of input/output places and transitions over which a net may interact with its environment, either by exchanging tokens over open places, or by synchronizing on visible transitions. 
OPNs with open places \cite{wolf2009} are ideal for modeling asynchronous communicating systems as a composition of nets obtained by fusing their shared open places, representing interfaces. 
Portnets \cite{bera2011component} are a subclass of OPN that explicitly model the behavior of an interface by defining permissible interactions with its environment, i.e., the order in which tokens may be exchanged over open input and output places. 
Portnets are ideal to describe interface contracts in the context of component-based systems. 
An interface of a component providing a service describes expectations from its environment (other components) and obligations towards it by means of a portnet enacting the server side (sell side) of an interface contract. 
A component that consumes the service provided by another component acts as its environment and defines usage by means of a mirrored portnet enacting the client side (buy side) of an interface contract. 
The compatibility of a pair of portnets is defined by not only matching input and output labels of open places, but by additionally imposing restrictions on the net structure which guarantees the weak termination property in portnet compositions. 

This section introduces the class of \emph{labeled portnets} which extends portnets~\cite{bera2011component} with a labeling function and a composition operation over sets of labeled portnets.  A strict requirement for compatibility of a pair of portnets was for their skeletons to agree on isomorphism. 
This induces a mirroring relation between a server portnet and its client, which is quite restrictive in practice, since a client may only be interested in using a part of the server portnet. We relax the isomorphism requirement by the \emph{partial mirror} relation between portnet pairs.\smallskip

An OPN is a Petri net with a distinguished set of interface places, partitioned into input and output places, that represent the interface of the net. 
The interaction of an OPN with its environment results in tokens (representing messages) being exchanged over interface places. 
If a transition of an OPN is associated to an interface place, then it is either an input place or an output place, but not both. So a transition either represents a send, or a receive. The net obtained by discarding interface places is called its \emph{skeleton}. By imposing restrictions on the skeleton, many subclasses of OPN are obtained. 

\begin{definition}\label{def:OPN}
Let $P, I, O$ be pairwise disjoint sets.
An \emph{open Petri net} \emph{(OPN)} is a tuple $(P, I, O, T, F, \textit{init}, \textit{fin}, \mL, \mu)$, where
\begin{itemize}
\item $(P \cup I \cup O, T, F, \mL, \mu)$ is a Petri net,
\item $I$ is the set of input places; we require $\pre{I} = \emptyset$, 
\item $O$ is the set of output places; we require $\post{O} = \emptyset$,
\item for all $t \in T$ we have either $\pre{t} \cap I = \emptyset$ or $\post{t} \cap O = \emptyset$, 
\item $\textit{init} \subseteq P$ is the set of initial places and $\textit{fin} \subseteq P$ is the set of final places.
\end{itemize}
\end{definition}
We refer to the places in $I$ and $O$ as the interface places of the OPN.
The \emph{direction of communication} of a transition with respect to a place in an OPN $N$ is a function $\lambda:T \rightarrow \{\textit{send}, \textit{receive}, \tau\}$ defined as follows:
\[
\lambda(t) = \left \{ \begin{array}[c]{ll}
             \textit{send} & \text{if } \post{t} \cap O \neq \emptyset \\
             \textit{receive} & \text{if } \pre{t} \cap I \neq \emptyset \\
             \tau & \text{otherwise }
             \end{array} \right .
\]
That is, $\lambda(t) = \textit{send}$ if $t$ is connected to an output, whereas $\lambda(t) = \textit{receive}$ if $t$ is connected to an input; $\lambda(t) = \tau$ if neither is the case.
We say that $t$ is a \emph{communicating transition} iff $\lambda(t) \neq \tau$.
The \emph{skeleton} of $N$, denoted $\mathcal{S}(N)$, is an OPN $\mathcal{S}(N) = (P, \emptyset, \emptyset, T, F', \textit{init}, \textit{fin},\mL,\mu)$ with $F'= F \cap ((P \times T) \cup (T \times P))$.
If the skeleton is an S-net then we call it a state machine open Petri net (S-OPN). 
An OPN $N$ is called an open workflow net (OWN) if its skeleton is a \emph{WFN}. 
An S-OPN with a skeleton that is a \emph{WFN} is called a state machine open workflow net (S-OWN).
The \emph{closure} of a S-OWN $N$, denoted by $\textit{closure(N)}$, is an S-OPN obtained by extending $N$ with a fresh transition $t \notin T$ such that $\lambda(t) = \tau$, $\pre{t} = \textit{fin}$ and $\post{t} = \textit{init}$.
\begin{figure}[t]
    \centering
    \includegraphics[scale=0.55]{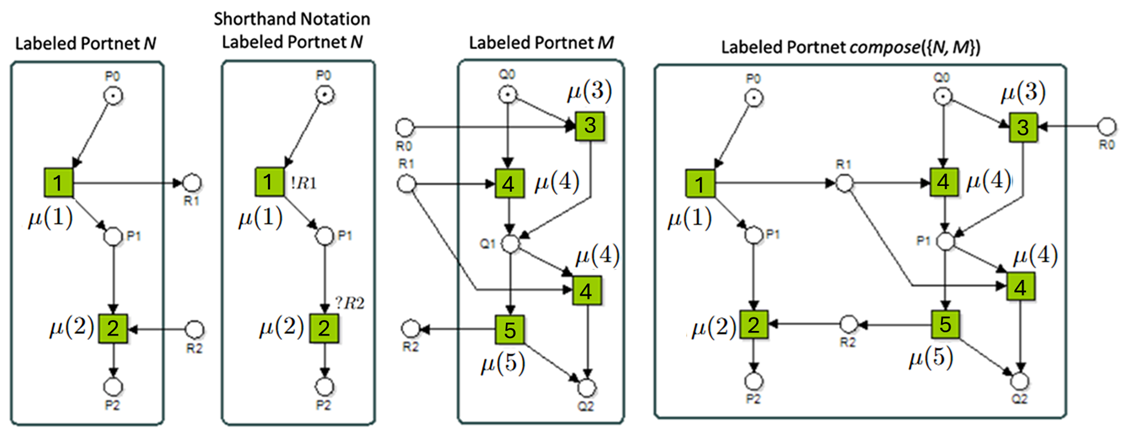}
    \caption{Labeled Portnets and their Composition}
\label{fig:labeledPortnetAndTheirComposition}
\end{figure}
A \emph{labeled portnet} is a subclass of OPN, characterized by additional restrictions on their structure. 
\begin{definition}
An OPN $N = (P, I, O, T, F, \textit{init}, \textit{fin}, \mL, \mu)$ is a \emph{labeled portnet} if:
\begin{itemize}
    \item $N$ is an S-OWN, and $|\textit{init}| = |\textit{fin}| = 1$, 
    \item all transitions $t \in T$ satisfy $|(\pre{t} \cup \post{t}) \cap (I \cup  O)| =1$, 
    \item for all $x \in I \cup O$ and all transitions $t, t' \in \pre{x} \cup \post{x}$ we have $\mu(t) = \mu(t')$, 
    \item for all transitions $t,t' \in T$ satisfying $\mu(t) = \mu(t')$ we have $(\pre{t} \cup \post{t}) \cap (I \cup O) = (\pre{t'} \cup \post{t'}) \cap (I \cup O)$, 
\end{itemize}
The net system $(N, m_0, m_f)$ of an OPN $N$ has as initial marking $m_0 = \textit{init}$ and final marking $m_f = \textit{fin}$, where sets $\textit{init}$ and $\textit{fin}$ are interpreted as bags. 
\end{definition}
\emph{Portnets} as defined in~\cite{bera2011component} are instances of labeled portnets where the labeling function $\mu$ satisfies $\mu(t) = t$.
That is, each transition in a portnet is labeled with its own, unique name.
To stress the distinction between portnets and labeled portnets, we sometimes refer to the former as \emph{ordinary portnets}.
\begin{lemma}
\label{lm:communication_direction}
Let $N$ be a labeled portnet and $t, t' \in T$ such that $\mu(t) = \mu(t')$. Then $\lambda(t) = \lambda(t')$.
\end{lemma}
\begin{proof}
Pick arbitrary $t,t' \in T$ satisfying $\mu(t) = \mu(t')$.
First observe that by the second requirement on labeled portnets, there is exactly one $x \in I \cup O$ such that $x \in \pre{t} \cup \post{t}$.
Suppose $x \in I$; the case $x \in O$ is fully dual.
Because by the second condition of an OPN, $\pre{I} = \emptyset$, we necessarily have $x \in \pre{t}$, so by definition $\lambda(t) = \textit{receive}$.
Since by the fourth requirement on labeled portnets, we find that $x \in \pre{t'} \cup \post{t'}$, either $x \in \pre{t'}$ or $x \in \post{t'}$. Suppose that $x \in \post{t'}$.
Then $t' \in \pre{I}$, contradicting condition $\pre{I} = \emptyset$ of an OPN. So $x \in \pre{t'}$, and hence $\lambda(t') = \textit{receive} = \lambda(t)$. 
\qed
\end{proof}

\noindent
Since the skeleton of a labeled portnet is a state machine workflow net, by Theorem~15~\cite{van2003soundness}, it is weakly terminating, as claimed by the corollary below. 
\begin{corollary}\label{cor:wt_skeleton}
The skeleton of a labeled portnet weakly terminates.
\end{corollary}
An example labeled portnet is shown in Fig.~\ref{fig:labeledPortnetAndTheirComposition}. As illustrated, the shorthand notation $?$ (receive) and $!$ (send) will be used to keep figures readable and when necessary the name of the connected interface places will be prefixed to it. Next to labels on transitions, numbers are used as their identifiers. Depending on the context, we may omit one of them. 

Two OPNs are \emph{composable} if the only nodes they share are interface places, and if they share a non-empty set of interface places then either the outputs of one are a subset of the inputs of the other, or their inputs and outputs are disjoint.
The \emph{composition} of a set of pairwise composable OPNs is almost a pairwise union of their corresponding tuples except that the shared interface places between OPNs become the internal places of the composition. 
Fig.~\ref{fig:labeledPortnetAndTheirComposition} shows a composition of two composable labeled portnets $N$ and $M$.
\begin{definition}\label{def:composability}
Two OPNs $N$ and $M$ are \emph{composable} iff 
\begin{itemize}
    \item $(P_N \cup I_N \cup O_N \cup T_N) \cap (P_M \cup I_M \cup O_M \cup T_M) = ((I_N \cup O_N) \cap (I_M \cup O_M))$, 

    \item If $(I_N \cap O_M) \cup (I_M \cap O_N) \neq \emptyset$, then also:
         \begin{itemize}
               \item $O_N \subseteq I_M$ and $O_M \subseteq I_N$, and
               \item $I_N \cap I_M = O_N \cap O_M = \emptyset$.
         \end{itemize}

\end{itemize}
The composition of a non-empty set $S$ of pairwise composable OPNs is an OPN $\textit{compose}(S) = (P, I, O, T, F, \textit{init}, \textit{fin},\mL, \mu)$, where
\begin{itemize}
\item $P = (\bigcup_{N \in S} P_N) \cup ( (\bigcup_{N \in S} I_N) \cap (\bigcup_{N \in S} O_N))$,
\item $I = (\bigcup_{N \in S} I_N) \setminus (\bigcup_{N \in S} O_N)$, $O = (\bigcup_{N \in S} O_N) \setminus (\bigcup_{N \in S} I_N)$,
\item $T = (\bigcup_{N \in S} T_N)$, $F = (\bigcup_{N \in S} F_N)$, 
\item $\textit{init} = (\bigcup_{N \in S} \textit{init}_N)$, $\textit{fin} = (\bigcup_{N \in S} \textit{fin}_N)$,
\item $\mathcal{L} = (\bigcup_{N \in S} \mathcal{L}_N)$ and $\mu = (\bigcup_{N \in S} \mu_N)$.
\end{itemize}
\end{definition}
We remark that since the shared places of OPNs become internal places when we compose them, in general for sets $S_1$ and $S_2$ of pairwise composable OPNs we do not have $\textit{compose}(S_1 \cup S_2) = \textit{compose}(S_1 \cup \textit{compose}(S_2))$.

\begin{figure}[t] 
    \centering
    \includegraphics[scale=0.55]{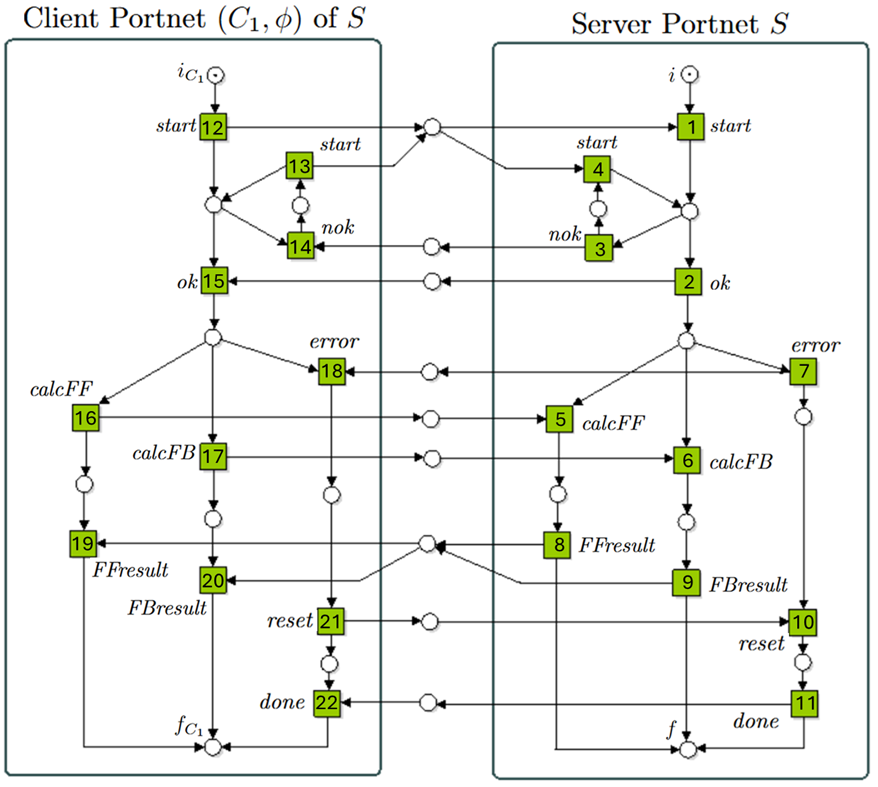}
    \caption{Deadlock in a Composition of Labeled Portnet and its Mirror}
    \label{fig:deadlocked}
\end{figure}
The \emph{partial mirror} relation relaxes the isomorphism requirement on skeletons of a pair of portnets \cite{bera2011component} by allowing one to drop some transitions with a \emph{send} direction, offering designers the option to exercise only the relevant parts of the other labeled portnet.

\begin{definition}\label{def:partial_mirror}
A \emph{partial mirror} of a labeled portnet $N$ is a pair $(M, \phi)$ where $M$ is a labeled portnet and $\phi: M \xrightarrow{} N$ an injective mapping satisfying:

\begin{itemize}
    \item $\phi(P_M) \subseteq P_N$, $\phi(T_M) \subseteq T_N$, 
    $\phi(I_M) \subseteq O_N$ and $\phi(O_M) \subseteq I_N$, 
    \item $F_M$ satisfies that for all $(x,y) \in F_M$:
    \begin{itemize}
        \item if $x \notin I_M$ and $y \notin O_M$, then $(\phi(x),\phi(y)) \in F_N$,
        \item if $x \in I_M$ or $y \in O_M$, then $(\phi(y),\phi(x)) \in F_N$,
    \end{itemize}
    \item $\phi(\textit{init}_M) = \textit{init}_N$ and $\phi(\textit{fin}_M) = \textit{fin}_N$,
    \item for all $t \in T_M$ we have $\mu_M(t) = \mu_N(\phi(t))$,
    \item for all places $p \in P_M$ and all transitions $t \in \post{\phi(p)}$ in $N$, if $\lambda(t) = \textit{send}$ then there is some $t' \in \post{p}$ such that 
    $\phi(t') = t$.
\end{itemize}
\end{definition}
We say that $(M,\phi)$ is a \emph{mirrored labeled portnet} in case $\phi$ is a bijection. 
The labeled portnet $M$ is referred to as the \emph{client portnet} of the \emph{server portnet} $N$.

A partial mirror relation by itself does not guarantee weak termination of the composition of a labeled portnet and its (partial) mirror. 
Fig.~\ref{fig:deadlocked} shows an example of such a composition w.r.t. $\phi$, where one of the deadlocks can be reached by firing the sequence of transitions $\langle 12, 1, 2, 15, 7, 16 \rangle$.

\section{Well-Formed Labeled Portnets}

In~\cite{bera2011component}, a notion of compatibility between pairs of ordinary portnets was introduced that guarantees weak termination of their composition.
This notion of compatibility requires the skeletons of the ordinary portnets to satisfy the \emph{choice} and \emph{leg} property; for the sake of completeness, we repeat these below.
These are structural properties of the ordinary portnets that can be checked efficiently.

The \emph{choice property} does not allow transitions in the postset of a place to present a race between an input and an output: either all such transitions perform a send, or receive, but not a mix of both. As a result it is not possible to model a race between a client and server portnet. 
The \emph{leg property} prevents a client and server portnet from getting confused and going down different paths. 
By requiring every path from a split (a place with postset greater than one, including initial place) to a join (a place with preset greater than one, including final place) to have at least two transitions with different communication direction, confusions are prevented. 
\begin{definition}[\textbf{Choice and Leg Property}]\label{def:choice_leg_property}
A portnet $N$ satisfies the \emph{choice property} if all $p \in P_N$ satisfy
that for all $t, t' \in \post{p}$ we have $\lambda(t) = \lambda(t')$. 
It satisfies \emph{leg property} if $\forall \beta = \langle p_1, t_1, \ldots, t_{n-1}, p_n \rangle \in PS(N): (|\post{p_1}| > 1 \vee p_1 = \textit{init}_N) \wedge (|\pre{p_n}| > 1 \vee p_n = \textit{fin}_N) : \exists t,t' \in \beta : \lambda(t) \neq \tau \wedge \lambda(t') \neq \tau \wedge \lambda(t) \neq \lambda(t')$. 
\end{definition}
This section presents a relaxation of these properties by introducing the 
\emph{observable choices}, the \emph{diamond} and \emph{loop} properties.
Moreover, we consider these properties in the more general setting of labeled portnets. 
As a consequence, the class of labeled portnets satisfying these properties encompasses a larger variety of interaction protocols, commonly encountered in practice.\smallskip

\begin{figure}[t]
    \centering 
    \includegraphics[scale=0.65]{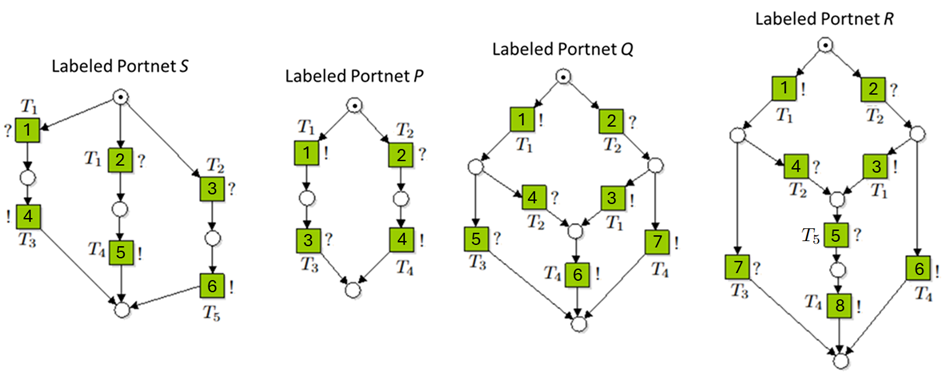}
    \caption{Portnet Properties}
    \label{fig:labeledPortnetProperties}
\end{figure}

A labeled portnet satisfies the \emph{observable choices} property if distinct transitions in the postset of a place have different labels. 
This enables one to decide, based on observing labels only, exactly which non-interface place of a labeled portnet becomes marked upon firing a sequence of transitions.
For instance, portnet $S$ in Fig.~\ref{fig:labeledPortnetProperties} violates this property, while the others satisfy it. 
\begin{definition}[\textbf{Observable Choices}]
A labeled portnet satisfies the observable choices property if all $p \in P$, and all distinct $t, t' \in \post{p}$ we have $\mu(t) \neq \mu(t')$. 
\end{definition}
The \emph{diamond property}, stated below, relaxes the aforementioned choice property by allowing a race between a client and server to initiate sending of messages to each other from a split place. 
Intuitively, such a race is benign if, whichever way the race is resolved, the option to handle the interaction that lost the race persists and a subsequent handling of that interaction still has the option to lead to an aligned view on each-other's state. 
This intuition is formalized below.

\begin{definition}[\textbf{Diamond Property}]
A labeled portnet satisfies the diamond property if all its places $p \in P$ 
and all $t,t' \in \post{p}$ with $\lambda(t) \neq \lambda(t')$ satisfy 
$$ 
\forall q \in \post{t},q' \in \post{t'}: \exists u \in \post{q}, u' \in \post{q'}: 
\post{u} \cap \post{u'} \neq \emptyset \wedge \mu(t) = \mu(u') \wedge \mu(t') = \mu(u)
$$
\end{definition}
\begin{property} If a labeled portnet satisfies the choice property then it also satisfies the diamond property.
\end{property}
Labeled portnet $S$, depicted in Fig.~\ref{fig:labeledPortnetProperties}, satisfies the choice property, while the other labeled portnets violate it. 
Labeled portnet $P$ violates the diamond property, while labeled portnets $Q$ and $R$ satisfy it.

\begin{definition}
For a given labeled portnet $N$ with place $p \in P$, transition $t$ and label $l \in \mL$, we define $\textit{pathset}(p,t,l)$ as follows:
\[
\begin{array}{ll}
\{ \langle p \rangle \circ \pi \in \textit{PS}(N) \mid &
\exists \sigma \in T^*, t'' \in T: \pi_{|T} = \langle t \rangle \circ \sigma \circ \langle t'' \rangle \wedge {} \\
& \mu(t'') = l \wedge \forall u \in T: u \in \sigma \Rightarrow \mu(u) \neq l
\}
\end{array}
\]
\end{definition}

\begin{figure}[t]
    \centering
    \includegraphics[scale=0.57]{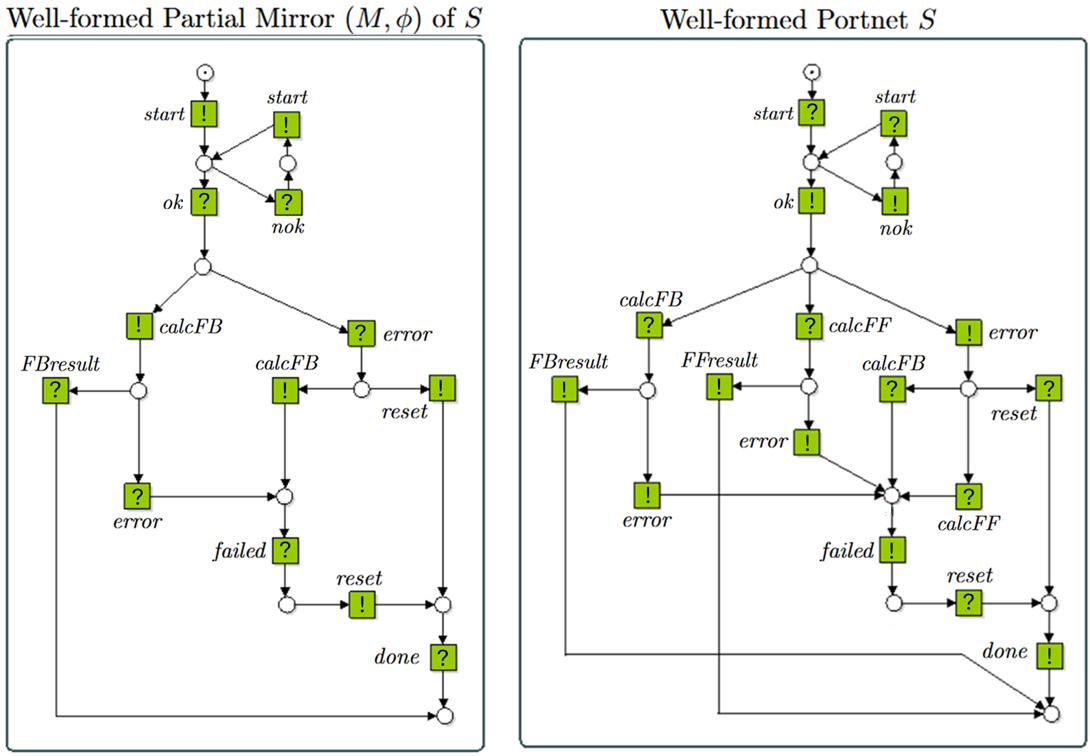}
    \caption{Well-formed portnet and its well-formed partial mirror portnet}
    \label{fig:wellFormedPortnets}
\end{figure}

The \emph{loop property}, formalized below, requires any path that starts with resolving competing receiving (resp.\ sending) interactions by handling one of them, to also handle the interaction with the other one, but not before a send (resp.\ receive) interaction occurred. 
The goal is to prevent confusion between client and server, causing both to go down different paths. For instance, consider send transitions $3$ and $7$ in the postset of a split in portnet $Q$ (Fig.~\ref{fig:labeledPortnetProperties}). 
Suppose $Q$ fires the sequence of \emph{send} transitions $\langle 3, 6 \rangle$. 
Then its partial mirror $(X,\phi)$ may choose to fire its \emph{receive} transition $t$, where $\phi(t) = 7$ (since $6$ and $7$ have the same label $T_4$, they are connected to the same interface place), resulting in leftover tokens in the interface place of transition $3$. 
Portnet $R$ resolves this by forcing a synchronization with \emph{receive} transition $5$. 
\begin{definition}[\textbf{Loop Property}]
A labeled portnet satisfies the loop property if all places $p \in P$ 
satisfy that all distinct transitions $t, t' \in \post{p}$ with $\lambda(t) = \lambda(t')$ and all paths $\sigma \in \hbox{pathset}(p,t,\mu(t'))$ have a transition $t'' \in \sigma$ with $\lambda(t'') \neq \lambda(t)$.

\end{definition}

\begin{definition}[\textbf{Well-Formed Portnet}]
A labeled portnet is \emph{well-formed} if it satisfies the observable choices, 
the diamond and the loop properties. 
\end{definition}
The portnet $S$ depicted in Fig.~\ref{fig:wellFormedPortnets} is a modification of portnet $S$, depicted in Fig.~\ref{fig:deadlocked}; the modifications ensure the portnet is well-formed. Its well-formed partial mirror $M$ drops the transition labeled \textit{calcFF} and as a consequence, also the conflict transitions \textit{FFResult} and \textit{error} to satisfy that the skeleton of $M$ remains a WFN.
\begin{figure}[t] 
    \centering
    \includegraphics[scale=0.29]{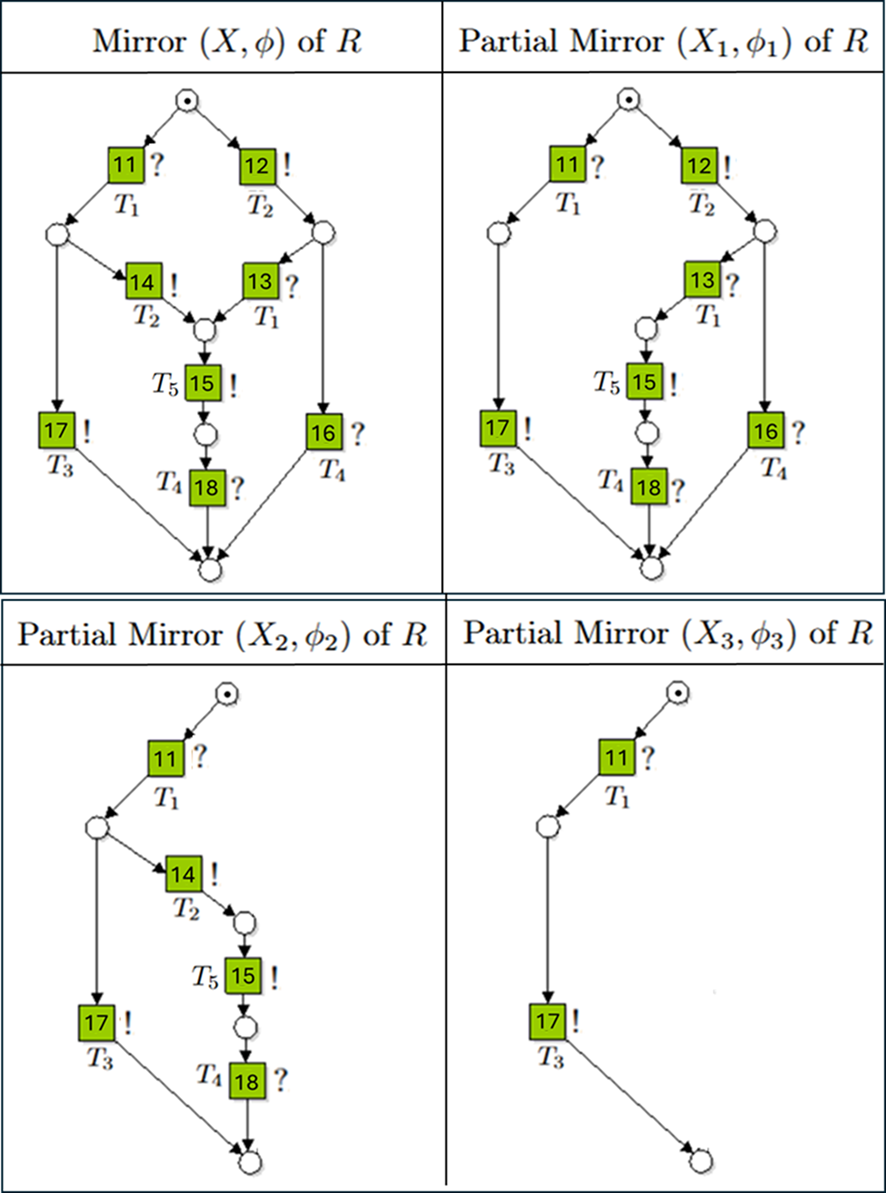}
    \caption{Partial mirrors of well-formed portnet $R$ in Fig.~\ref{fig:labeledPortnetProperties}}
    \label{fig:partialMirrorsNotWellFormed}
\end{figure}
We remark that well-formedness of labeled portnets is not preserved by partial mirroring.
To see this, consider the well-formed portnet $R$, depicted in Fig.~\ref{fig:labeledPortnetProperties} and its partial mirrors in Fig.~\ref{fig:partialMirrorsNotWellFormed}, where $\phi(11) = \phi_1(11) = \phi_2(11) = \phi_3(11) = 1$, 
$\phi(12) = \phi_1(12) = 2$, 
$\phi(13) = \phi_1(13) = 3$, 
$\phi(14) = \phi_2(14) = 4$, 
$\phi(15) = \phi_1(15) = \phi_2(15) = 5$, 
$\phi(16) = \phi_1(16) = 6$, 
$\phi(17) = \phi_1(17) = \phi_2(17) = \phi_3(17) = 7$, and $\phi(18) = \phi_1(18) = \phi_2(18) = 8$.
The partial mirror $(X_1,\phi_1)$ is not well-formed since it violates the diamond property. However, the two other partial mirrors $(X_2,\phi_2)$ and $(X_3,\phi_3)$ are well-formed. The mirrored labeled portnet $(X,\phi)$ is also well-formed. 
\section{Weak Termination of Composition of Labeled Portnets}
As demonstrated by the example in Fig.~\ref{fig:deadlocked}, the composition of a labeled portnet and its partial mirror does not necessarily lead to a weakly terminating OPN.
In this section, we show that weak termination of a composition of two labeled portnets $N$ and $M$ \emph{is} guaranteed when both are well-formed and $M$ is a partial mirror of $N$. 
Indeed, the main theorem of this section is the following:
\begin{theorem}\label{thm:wt_portnets}
Let $N$ and $M$ be well-formed labeled portnets and let $(M,\phi)$ be a partial mirror of $N$. 
Let $S = \textit{compose}(\{N,M\})$ and $i_S = i_N + i_M$ be the initial marking of $S$ and let $f_S = f_N + f_M$ be the final marking of $S$.
Then the net system $(S, i_S, f_S)$ weakly terminates.
\end{theorem}

In the remainder of this section we break down the correctness argument of the theorem before ultimately presenting its proof.

Fix a well-formed labeled portnet $N$---from hereon the \emph{server}---and a well-formed partial mirror $(M,\phi)$ of $N$---from hereon the \emph{client}.
Suppose that in their interactions, at some point the server wishes to execute a send transition.
The client may still be catching up with previously executed send transitions, due to the asynchronous nature of their communication.
However, due to the diamond property, it does not matter when exactly the server fires its send transition, since it will have the same effect on the overall state of the system.
The following lemma formalizes this phenomenon. 
\begin{lemma}\label{lem:diamond_path}
Let $N$ and $M$ be well-formed labeled portnets and let $(M,\phi)$ be a partial mirror of $N$.  
Let $S = \textit{compose}(\{N,M\})$ and $i_S = i_N + i_M$ be the initial marking of $S$ and $f_S = f_N + f_M$ be the final marking of $S$.
For any marking $m, m', m_0,\dots,m_n \in \mathcal{R}(S,i_S)$, and all transitions $t,t_1,\dots,t_n$ of $N$ such that $m = m_0 \xrightarrow{t_1} \dots \xrightarrow{t_n} m_n$, where each $\lambda(t_i) = receive$, and $m \xrightarrow{t} m'$ with $\lambda(t) = send$, there exists markings $m_0',\dots,m_n' \in \mathcal{R}(S,i_S)$ and transitions $\bar{t},t_1',\dots,t_n'$ of $N$ such that $m' =  m_0' \xrightarrow{t_1'} \dots \xrightarrow{t_n'} m_n'$ and $m_n \xrightarrow{\bar{t}} m_n'$ and 
\begin{enumerate}
\item $\mu(t_i) = \mu(t_i')$ and $\lambda(t_i) = \lambda(t_i')$;
\item $\mu(t) = \mu(\bar{t})$ and $\lambda(t) = \lambda(\bar{t})$.
\end{enumerate}
\end{lemma}
\begin{proof}
We use an induction on the length of the sequence $t_1\dots t_n$ of receive transitions in portnet $N$ (the case where these transitions are from $M$ is analogous).

\begin{itemize}
\item Base case, $n = 0$. Then there is no receive transition and the statement holds vacuously.

\item Inductive step. Suppose the statement holds true of a sequence $m_0 \xrightarrow{t_1}\dots\xrightarrow{t_n} m_n$.
Suppose that $m_n \xrightarrow{t_{n+1}} m_{n+1}$ for some receive transition $t_{n+1}$, fired by portnet $N$, and $m_0 \xrightarrow{t} m_0'$, with $\lambda(t) = send$.
By the induction hypothesis, we know that $m_0' \xrightarrow{t_1'} m_1' \dots\xrightarrow{t_n'} m_n'$, with $\mu(t_i) = \mu(t_i')$ and $\lambda(t_i) = \lambda(t_i')$.
Moreover, $m_n \xrightarrow{\bar{t}} m_n'$ for some $\bar{t}$ satisfying $\mu(t) = \mu(\bar{t})$ and $\lambda(\bar{t}) = send$.
Since $m_n \xrightarrow{t_{n+1}} m_{n+1}$ and $m_n \xrightarrow{\bar{t}} m_n'$, we find that $\pre{t_{n+1}} \cap \pre{\bar{t}} \neq \emptyset$.
By the diamond property, we then find that there must be some $\tilde{t}$ and $m_{n+1}'$ such that $m_{n+1} \xrightarrow{\tilde{t}} m_{n+1}'$ and some $t'_{n+1}$ such that $m_n' \xrightarrow{t'_{n+1}} m_{n+1}'$ and $\mu(\tilde{t}) = \mu(\bar{t})$ and $\mu(t'_{n+1}) = \mu(t_{n+1})$, finishing the proof. 
\qed
\end{itemize}
\end{proof}
Similar to the exchange lemma (see~\cite{desel1995free}), we show that an enabled transition can be fired anywhere in a firing sequence, as long as the transition does not share any place with the preset of transitions in that executable firing sequence. 
\begin{lemma}
\label{shuffle_lemma}
Let $(N,m_0,m_f)$ be  a net system, with $N$ a labeled OPN.
Suppose $t \in T$ and $\sigma \in T^*$ are such that for all $t' \in \sigma$, $\pre{t'}$ is disjoint with $\pre{t} \cup \post{t}$.
Then for every $m \in \mathcal{R}(N,m_0,m_f)$ satisfying $m \xrightarrow{t}$ and $m \xrightarrow{\sigma}$ there is some $m'$ such that for every $\sigma' \in (\langle t \rangle || \sigma)$ we have $m \xrightarrow{\sigma'} m'$.
\end{lemma}
\begin{proof}
By the marking equation lemma \cite{desel1995free}, all possible sequences in $(\langle t \rangle ||\sigma)$ lead to the same resulting marking (provided the number of occurrences of transitions are the same). Hence, we only need to show that any sequence in $(\langle t \rangle ||\sigma)$ is an executable firing sequence. 
Let $\alpha'$ and $\alpha''$ such that $\alpha = \alpha' \circ \langle t \rangle \circ \alpha''$, i.e., $\sigma = \alpha' \circ \alpha''$ and $\alpha \in (\langle t \rangle || \sigma)$. Hence $m \xrightarrow{\alpha'} \tilde{m} \xrightarrow{\alpha''} \tilde{m}'$ for some $\tilde{m}, \tilde{m}' \in \mathcal{R}(N,m)$. As $\pre{t}$ and $\pre{\alpha}'$ are disjoint, transition $t$ is enabled in $\tilde{m}$. Consequently, a marking $\tilde{m}'' \in \mathcal{R}(N,m)$ exists with $\tilde{m} \xrightarrow{t} \tilde{m}''$. 
Since $\post{t}$ and $\pre{\alpha}''$ are disjoint, 
$\tilde{m}'' \xrightarrow{\alpha''} m'$, which proves the lemma. 
\qed
\end{proof}
A receive transition of a labeled portnet and a receive transition of its partial mirror do not share an interface place. 
So in a composition of a well-formed labeled portnet and its well-formed partial mirror, a sequence of receive transitions of the former may be shuffled with the latter, leading to the same marking. 
\begin{lemma}\label{lem:swapping}
Let $N$ and $M$ be well-formed labeled portnets and let $(M,\phi)$ be a partial mirror of $N$. 
Let $S = \textit{compose}(\{N,M\})$ and $i_S = i_N + i_M$ be the initial marking of $S$ and $f_S = f_N + f_M$ be the final marking of $S$. 
For any reachable marking $m_0,\dots,m_n$, and all transitions $t_1,\dots,t_n \in T_N$ and $u_1,\dots,u_k \in T_M$ such that
$m_0 \xrightarrow{t_1} \dots \xrightarrow{t_n} m_n \xrightarrow{u_1} \dots \xrightarrow{u_k} m_{n+k}$ and $\lambda(t_i) = \lambda(u_j) = receive$ for all relevant $i,j$, there are markings $m_0',\dots, m_{n+k}'$ such that $m_0 = m_0' \xrightarrow{u_1} \dots \xrightarrow{u_k} m_k' \xrightarrow{t_1} \dots \xrightarrow{t_n} m_{n+k}' = m_{n+k}$.
\end{lemma}
\begin{proof}
Follows from Lemma~\ref{shuffle_lemma} since receive transitions of a well formed labeled portnet and its partial mirror do not share interface places with each other. 
\qed
\end{proof}
From a reachable marking of a composition of a well-formed server portnet and its well-formed partial mirror client, firing an enabled sequence of receive transitions of the server cannot disable an enabled send transition of a client since it does not consume from an interface place. 
\begin{lemma}\label{lem:swap_send_receive}
Let $N$ and $M$ be well-formed labeled portnets and let $(M,\phi)$ be a partial mirror of $N$. 
Let $S = \textit{compose}(\{N,M\})$ and $i_S = i_N + i_M$ be the initial marking of $S$ and $f_S = f_N + f_M$ be the final marking of $S$. 
For any reachable marking $m_0,\dots,m_n,m_0'$, and all transitions $t_1,\dots,t_n \in T_N$ with $\lambda(t_i) = receive$, and $u \in T_M$ with $\lambda(u) = send$ such that
$m_0 \xrightarrow{t_1} \dots \xrightarrow{t_n} m_n$ and $m_0 \xrightarrow{u} m_0'$, there are markings $m_1',\dots, m_n'$ such that $m_0' \xrightarrow{t_1} \dots \xrightarrow{t_n} m_n'$ and $m_n \xrightarrow{u} m_n'$.
\end{lemma}
\begin{proof}
Since the send transition $u$ does not share its preset with any of the receive transition $t_1 \ldots t_n$, the transition $u$ remains enabled in marking $m_n$. 
\qed
\end{proof}
Building on the previous lemmas, the next proposition shows that for any reachable marking in a composition of a well-formed server portnet and its well-formed partial mirror client, there exists a sequence of receive transitions that can be executed first in the server, then in the client, leading to a marking where a place in the client and its mirror are marked, and all other places are empty. 
This means, informally, that server and client can catch up with each other and align their view on each other's state by processing receive transitions only.
\begin{proposition}\label{lemma:synchronisable}
Let $N$ and $M$ be well-formed labeled portnets and let $(M,\phi)$ be a partial mirror of $N$. 
Let $S = \textit{compose}(\{N,M\})$ and $i_S = i_N + i_M$ be the initial marking of $S$ and $f_S = f_N + f_M$ be the final marking of $S$.
For any marking $m_0,\dots,m_n$, and all transitions $t_1,\dots,t_n$ such that
$i_S = m_0 \xrightarrow{t_1} \dots \xrightarrow{t_n} m_n$, there are markings $m_{n+1},\dots, m_{n+k}$ and transitions $t_{n+1},\dots, t_{n+k}$ such that $m_{n} \xrightarrow{t_{n+1}} \dots \xrightarrow{t_{n+k}} m_{n+k}$ and:
\begin{enumerate}
\item for all $i \geq n+1$, we have $\lambda(t_i) = receive$;
\item $t_{n+1} \dots t_{n+k} \in T_N^*T_M^*$; \emph{i.e.\/} the transitions first execute in $N$ and then in $M$;
\item there is a place $p \in P_M$ such that $m_{n+k}(p) = m_{n+k}(\phi(p)) = 1$ and $m_{n+k}(p') = 0$ for all $p' \notin P_N \cup P_M$; \emph{i.e.\/} all interface places are empty and there is a token in both place $p$ of $M$ and its mirror $\phi(p)$ of $N$. 
\end{enumerate}
\end{proposition}
\begin{proof}
We use an induction on the length of the path $m_0\dots m_n$.
\begin{itemize}
\item Base case, $n = 0$. Then there is only one marking, \emph{viz.\/} $m_0$, and no transition has fired. 
Then the empty sequence of markings (taking $k = 0$), satisfies the three desired conditions vacuously.

\item Inductive step. Consider markings $m_0,\dots, m_{n+1}$ and transitions $t_1,\dots,t_{n+1}$ such that $i_S = m_0 \xrightarrow{t_1} \dots \xrightarrow{t_{n+1}} m_{n+1}$, and assume that for $i_S = m_0 \xrightarrow{t_1} \dots \xrightarrow{t_n} m_n$, there are markings $m_{n+1}',\dots, m_{n+k}'$ and transitions $t_{n+1}',\dots, t_{n+k}'$ such that $m_{n} \xrightarrow{t_{n+1}'} \dots \xrightarrow{t_{n+k}'} m_{n+k}'$ satisfying the three desired properties.
Then, by the Induction Hypothesis, the situation is as follows:

\begin{tikzpicture}[node distance=3em]
\footnotesize
\newlength{\mywidth}
\newlength{\myheight}
\node (m0) {$m_0$};
\settowidth{\mywidth}{$\dots$}
\node[right of=m0,xshift=\the\mywidth] (mdots) {$\dots$};
\settowidth{\mywidth}{$m_{n}$}
\node[right of=mdots,xshift=\the\mywidth] (mn) {$m_n$};
\settowidth{\myheight}{$t_{n+1}$}
\node[below of=mn,yshift=-\the\mywidth] (mn_p_1) {$m_{n+1}$};
\settowidth{\mywidth}{$t_{n+1}$}
\node[right of=mn,xshift=\the\mywidth] (mdots') {$\dots$};
\settowidth{\mywidth}{$t'_{n+2}$}
\node[right of=mdots',xshift=\the\mywidth] (mn_p_l') {$m'_{n+l}$};
\settowidth{\mywidth}{$t'_{n+l+1}$}
\node[right of=mn_p_l',xshift=\the\mywidth] (mdots_') {$\dots$};
\settowidth{\mywidth}{$t'_{n+k}$}
\node[right of=mdots_',xshift=\the\mywidth] (mk') {$m_{n+k}'$};

\draw[->] 
  (m0) edge node[above] {$t_1$} (mdots)
  (mdots) edge node[above] {$t_n$} (mn)
  (mn) edge node[right] {$t_{n+1}$} (mn_p_1)
  (mn) edge node[above] {$t_{n+1}'$} (mdots')
  (mdots') edge node[above] {$t'_{n+l}$} (mn_p_l')
  (mn_p_l') edge node[above] {$t'_{n+l+1}$} (mdots_')
  (mdots_') edge node[above] {$t_{n+k}'$} (mk')
;
\end{tikzpicture}

With $\lambda(t_{n+1}') = \dots = \lambda(t_{n+k}') = receive$, and $m'_{n+k}(p) = m'_{n+k}(\phi({p})) = 1$ for some (non-interface) place $p \in P_M$ and 
for all $p \notin P_N \cup P_M$, $m'_{n+k}(p) = 0$. 
Moreover, we have $t'_{n+1}, \dots, t'_{n+l} \in T_N$ and $t'_{n+l+1}, \dots, t'_{n+k} \in T_M$.\\

We distinguish two cases, based on the communication direction of $t_{n+1}$.
\begin{itemize}
\item Case $\lambda(t_{n+1}) = send$. 
We make a further case distinction based on the portnets firing $t_{n+1}$ and $t_{n+1}'$.

\begin{itemize}
\item both $t_{n+1}$ and $t_{n+1}'$ belong to $T_N$.
By Lemma~\ref{lem:diamond_path}, there is a path $m_{n+1} \xrightarrow{t_{n+1}''} \dots \xrightarrow{t_{n+l}''} m_{n+l+1}$ such that for all $1 \leq i \leq l$, $\mu(t_{n+i}'') = \mu(t_{n+i}')$ and $\lambda(t_{n+i}'') = \lambda(t_{n+i}'')$, and $m'_{n+l} \xrightarrow{\tilde{t}_{n+1}} m_{n+l+1}$ for some $\tilde{t}_{n+1}$ satisfying $\mu(\tilde{t}_{n+1}) = \mu(t_{n+1})$ and $\lambda(\tilde{t}_{n+1})= send$:

\begin{tikzpicture}[node distance=3em]
\footnotesize
\node (m0) {$m_0$};
\settowidth{\mywidth}{$\dots$}
\node[right of=m0,xshift=\the\mywidth] (mdots) {$\dots$};
\settowidth{\mywidth}{$m_{n}$}
\node[right of=mdots,xshift=\the\mywidth] (mn) {$m_n$};
+\settowidth{\myheight}{$t_{n+1}$}
\node[below of=mn,yshift=-\the\mywidth] (mn_p_1) {$m_{n+1}$};
\settowidth{\mywidth}{$t_{n+1}$}
\node[right of=mn,xshift=\the\mywidth] (mdots') {$\dots$};
\settowidth{\mywidth}{$t'_{n+2}$}
\node[right of=mdots',xshift=\the\mywidth] (mn_p_l') {$m'_{n+l}$};
\settowidth{\mywidth}{$t'_{n+l+1}$}
\node[right of=mn_p_l',xshift=\the\mywidth] (mdots_') {$\dots$};
\settowidth{\mywidth}{$t'_{n+k}$}
\node[right of=mdots_',xshift=\the\mywidth] (mk') {$m_{n+k}'$};
\settowidth{\mywidth}{$t_{n+1}$}
\node[right of=mn_p_1,xshift=\the\mywidth] (mn_p_2) {$\dots$};
\settowidth{\mywidth}{$t'_{n+2}$}
\node[right of=mn_p_2,xshift=\the\mywidth] (mn_p_l_p_1) {$m_{n+l+1}$};

\draw[->] 
  (m0) edge node[above] {$t_1$} (mdots)
  (mdots) edge node[above] {$t_n$} (mn)
  (mn) edge node[right] {$t_{n+1}$} (mn_p_1)
  (mn) edge node[above] {$t_{n+1}'$} (mdots')
  (mdots') edge node[above] {$t'_{n+l}$} (mn_p_l')
  (mn_p_l') edge node[above] {$t'_{n+l+1}$} (mdots_')
  (mdots_') edge node[above] {$t_{n+k}'$} (mk')
  (mn_p_1) edge node[above] {$t_{n+1}''$} (mn_p_2)
  (mn_p_2) edge node[above] {$t''_{n+l}$} (mn_p_l_p_1)
  (mn_p_l') edge[dotted] node[right] {$\tilde{t}_{n+1}$} (mn_p_l_p_1)
;
\end{tikzpicture}

By lemma~\ref{lem:swap_send_receive}, it then also follows that we can extend the execution:

\begin{tikzpicture}[node distance=3em]
\footnotesize
\node (m0) {$m_0$};
\settowidth{\mywidth}{$\dots$}
\node[right of=m0,xshift=\the\mywidth] (mdots) {$\dots$};
\settowidth{\mywidth}{$m_{n}$}
\node[right of=mdots,xshift=\the\mywidth] (mn) {$m_n$};
\settowidth{\myheight}{$t_{n+1}$}
\node[below of=mn,yshift=-\the\mywidth] (mn_p_1) {$m_{n+1}$};
\settowidth{\mywidth}{$t_{n+1}$}
\node[right of=mn,xshift=\the\mywidth] (mdots') {$\dots$};
\settowidth{\mywidth}{$t'_{n+2}$}
\node[right of=mdots',xshift=\the\mywidth] (mn_p_l') {$m'_{n+l}$};
\settowidth{\mywidth}{$t'_{n+l+1}$}
\node[right of=mn_p_l',xshift=\the\mywidth] (mdots_') {$\dots$};
\settowidth{\mywidth}{$t'_{n+k}$}
\node[right of=mdots_',xshift=\the\mywidth] (mk') {$m_{n+k}'$};
\settowidth{\mywidth}{$t_{n+1}$}
\node[right of=mn_p_1,xshift=\the\mywidth] (mn_p_2) {$\dots$};
\settowidth{\mywidth}{$t'_{n+2}$}
\node[right of=mn_p_2,xshift=\the\mywidth] (mn_p_l_p_1) {$m_{n+l+1}$};
\settowidth{\mywidth}{$t'_{n+l+1}$}
\node[right of=mn_p_l_p_1,xshift=\the\mywidth] (_mdots') {$\dots$};
\settowidth{\mywidth}{$t'_{n+k}$}
\node[right of=_mdots',xshift=\the\mywidth] (mk_p_1) {$m_{n+k+1}$};

\draw[->] 
  (m0) edge node[above] {$t_1$} (mdots)
  (mdots) edge node[above] {$t_n$} (mn)
  (mn) edge node[right] {$t_{n+1}$} (mn_p_1)
  (mn) edge node[above] {$t_{n+1}'$} (mdots')
  (mdots') edge node[above] {$t'_{n+l}$} (mn_p_l')
  (mn_p_l') edge node[above] {$t'_{n+l+1}$} (mdots_')
  (mdots_') edge node[above] {$t_{n+k}'$} (mk')
  (mn_p_1) edge node[above] {$t_{n+1}''$} (mn_p_2)
  (mn_p_2) edge node[above] {$t''_{n+l}$} (mn_p_l_p_1)
  (mn_p_l') edge[dotted] node[right] {$\tilde{t}_{n+1}$} (mn_p_l_p_1)
  (mn_p_l_p_1) edge node[above] {$t'_{n+l+1}$} (_mdots')
  (_mdots') edge node[above] {$t'_{n+k}$} (mk_p_1)
  (mk') edge[dotted] node[right] {$\tilde{t}_{n+1}$} (mk_p_1)
;
\end{tikzpicture}

It now suffices to argue that there is a transition $t'_{n+k+1}$ such that $\mu(t'_{n+k+1}) = \mu(t_{n+1})$ and $\lambda(t'_{n+k+1}) = receive$ that is enabled. 
By the induction hypothesis, we find that there must be some place $p \in P_M$ and mirror $\phi({p}) \in P_N$ such that $m'_{n+k}(p) = m'_{n+k}(\phi({p})) = 1$. 
Let $p \in P_M$ and $\bar{p} = \phi(p) \in P_N$ be such.
Note that in the top execution sequence, we find that $N$ no longer executes transitions in marking $m'_{n+l}$, so $m'_{n+l}(q) = m'_{n+k}(q)$ for all $q \in P_N$.

Since $m'_{n+l}(\bar{p}) = 1$, $N$ is a S-OWN, and $m'_{n+l} \xrightarrow{\tilde{t}_{n+1}}$, we must have $\bar{p} \in \pre{\tilde{t}_{n+1}}$. 
Let $p' \in \post{\tilde{t}_{n+1}} \cap P_N$, and let $\tilde{\tilde{t}}_{n+1} \in T_M$ be such that $\phi(\tilde{\tilde{t}}_{n+1}) = \tilde{t}_{n+1}$. Note that this $\tilde{\tilde{t}}_{n+1} \in T_M$ exists because $\lambda(\tilde{t}_{n+1}) = \textit{send}$, $(\phi(p),\tilde{t}_{n+1}) \in F_N$ and $M$ is a partial mirror (by Def.~\ref{def:partial_mirror}). Then $\lambda(\tilde{\tilde{t}}_{n+1}) = \textit{receive}$. 
Since 
$\bar{p} \in \pre{\tilde{t}_{n+1}}$ and $\phi(\tilde{\tilde{t}}_{n+1}) = \tilde{t}_{n+1}$, we have $p \in \pre{\tilde{\tilde{t}}_{n+1}}$
and 
$\pre{\tilde{\tilde{t}}_{n+1}} \subseteq \post{\tilde{t}_{n+1}} \cup \{p\}$, we may execute $m_{n+k+1} \xrightarrow{\tilde{\tilde{t}}_{n+1}} m_{n+k+2}$ such that $m_{n+k+2}(p') = m_{n+k+2}(\bar{p}') = 1$, for $\bar{p}' \in \post{\tilde{\tilde{t}}_{n+1}} \cap P_M$ (Note $\bar{p}'$ exists since skeleton of a mirror is a workflow net - Def.~\ref{def:partial_mirror}) such that $\phi(\bar{p}') = p'$, and 
$m_{n+k+2}(i) = 0$ for all $i \notin \{p',\bar{p}'\}$.

\item $t_{n+1}$ belongs to $T_N$, but $t_{n+1}'$ belongs to $T_M$.
This case is a simplification of the previous case, since by the induction hypothesis we then have $t_{n+1}', \dots, t_{n+k}' \in T_M$.

\item $t_{n+1}$ belongs to $T_M$, but $t_{n+1}'$ belongs to $T_N$.
By Lemma~\ref{lem:swapping}, we find that we can transform:

\begin{tikzpicture}[node distance=3em]
\footnotesize
\node (m0) {$m_0$};
\settowidth{\mywidth}{$\dots$}
\node[right of=m0,xshift=\the\mywidth] (mdots) {$\dots$};
\settowidth{\mywidth}{$m_{n}$}
\node[right of=mdots,xshift=\the\mywidth] (mn) {$m_n$};
\settowidth{\myheight}{$t_{n+1}$}
\node[below of=mn,yshift=-\the\mywidth] (mn_p_1) {$m_{n+1}$};
\settowidth{\mywidth}{$t_{n+1}$}
\node[right of=mn,xshift=\the\mywidth] (mdots') {$\dots$};
\settowidth{\mywidth}{$t'_{n+2}$}
\node[right of=mdots',xshift=\the\mywidth] (mn_p_l') {$m'_{n+l}$};
\settowidth{\mywidth}{$t'_{n+l+1}$}
\node[right of=mn_p_l',xshift=\the\mywidth] (mdots_') {$\dots$};
\settowidth{\mywidth}{$t'_{n+k}$}
\node[right of=mdots_',xshift=\the\mywidth] (mk') {$m_{n+k}'$};

\draw[->] 
  (m0) edge node[above] {$t_1$} (mdots)
  (mdots) edge node[above] {$t_n$} (mn)
  (mn) edge node[right] {$t_{n+1}$} (mn_p_1)
  (mn) edge node[above] {$t_{n+1}'$} (mdots')
  (mdots') edge node[above] {$t'_{n+l}$} (mn_p_l')
  (mn_p_l') edge node[above] {$t'_{n+l+1}$} (mdots_')
  (mdots_') edge node[above] {$t_{n+k}'$} (mk')
;
\end{tikzpicture}

to:

\begin{tikzpicture}[node distance=3em]
\footnotesize
\node (m0) {$m_0$};
\settowidth{\mywidth}{$\dots$}
\node[right of=m0,xshift=\the\mywidth] (mdots) {$\dots$};
\settowidth{\mywidth}{$m_{n}$}
\node[right of=mdots,xshift=\the\mywidth] (mn) {$m_n$};
\settowidth{\myheight}{$t_{n+1}$}
\node[below of=mn,yshift=-\the\mywidth] (mn_p_1) {$m_{n+1}$};
\settowidth{\mywidth}{$t_{n+1}$}
\node[right of=mn,xshift=\the\mywidth] (mdots') {$\dots$};
\settowidth{\mywidth}{$t'_{n+2}$}
\node[right of=mdots',xshift=\the\mywidth] (mn_p_l') {$m'_{n+l}$};
\settowidth{\mywidth}{$t'_{n+l+1}$}
\node[right of=mn_p_l',xshift=\the\mywidth] (mdots_') {$\dots$};
\settowidth{\mywidth}{$t'_{n+k}$}
\node[right of=mdots_',xshift=\the\mywidth] (mk') {$m_{n+k}'$};

\draw[->] 
  (m0) edge node[above] {$t_1$} (mdots)
  (mdots) edge node[above] {$t_n$} (mn)
  (mn) edge node[right] {$t_{n+1}$} (mn_p_1)
  (mn) edge node[above] {$t_{n+l+1}'$} (mdots')
  (mdots') edge node[above] {$t'_{n+k}$} (mn_p_l')
  (mn_p_l') edge node[above] {$t'_{n+1}$} (mdots_')
  (mdots_') edge node[above] {$t_{n+l}'$} (mk')
;
\end{tikzpicture}

But then the argument boils down to the first case again.
The sequence of transitions that is then obtained can then again be rearranged using Lemma~\ref{lem:swapping}, yielding a sequence of transitions in $T_N^*T_M^*$.

\item $t_{n+1}$ and $t_{n+1}'$ both belong to $T_M$.
This case is again a simplification of the previous case.

\end{itemize}

\item Case $\lambda(t_{n+1}) = receive$. We again consider cases based on the portnet firing $t_{n+1}$ and $t_{n+1}'$.

\begin{itemize}

\item $t_{n+1}$ and $t_{n+1}'$ both belong to $T_N$.
We show that $t_{n+1} = t'_{n+1}$, from which we immediately reach the desired conclusion.
Towards a contradiction, assume that $t_{n+1} \neq t'_{n+1}$.
Since both $t_{n+1}$ and $t_{n+1}'$ are enabled, and $N$ and $M$ are S-Nets, we must have $\pre{t_{n+1}} \cap \pre{t_{n+1}'} \neq \emptyset$.
Moreover, the interface places connected to both $t_{n+1}$ and $t_{n+1}'$ must contain at least one token.
By the induction hypothesis, we know that in marking $m'_{n+k}$, there is no interface place containing a token.
Consequently, at least one of the transitions $t \in \{ t'_{n+2},\dots,t'_{n+k}\}$ must be such that $\mu(t) = \mu(t_{n+1})$.
But by the loop property, this cannot be the case, as a transition $t'$, with $\lambda(t') = send$ must fire before firing any transition with label $\mu(t_{n+1})$.
Contradiction, so $t_{n+1} = t'_{n+1}$. 

\item $t_{n+1}$ belongs to $T_N$, but $t'_{n+1}$ belongs to $T_M$.
This cannot be the case, since this contradicts the assumption that the sequence of transitions $t_{n+1}' \dots t'_{n+k}$, which only involves $M$ firing, empties all interface places.
However, receive transitions of $M$ cannot empty the interface place connected to $t_{n+1}$. 

\item $t_{n+1}$ belongs to $T_M$, but $t'_{n+1}$ belongs to $T_N$.
Using Lemma~\ref{lem:swapping}, we can swap the sequence of transitions $t'_{n+1} \dots t'_{n+l}$ and $t'_{n+l+1} \dots t'_{n+k}$.
Following the line of reasoning of the first case, we find that $t'_{n+l+1} = t_{n+1}$.
Another application of Lemma~\ref{lem:swapping} then yields the desired conclusion.

\item $t_{n+1}$ and $t_{n+1}'$ both belongs to $T_M$.
This case is a simplification of the previous case.
\qed
\end{itemize}
\end{itemize}
\end{itemize}
\end{proof}
\newcommand{\sendreceive}{\mathop{\langle\rangle}}
From a reachable marking in a composition of a well-formed server portnet and its well-formed partial mirror client such that a place in the client and its mirror are marked, and all other places empty, an alternating sequence of send and mirrored receive transitions can be fired leading to the final marking. The proof relies on weak termination of portnet skeletons (Corollary~\ref{cor:wt_skeleton}).
\begin{proposition}\label{lemma:homing_sequence}
Let $N$ and $M$ be well-formed labeled portnets and let $(M,\phi)$ be a partial mirror of $N$. 
Let $S = \textit{compose}(\{N,M\})$ and $i_S = i_N + i_M$ be the initial marking of $S$ and $f_S = f_N + f_M$ be the final marking of $S$.
Assume the skeleton of $M$ is weakly terminating. Let $m \in \mathcal{R}(S,i_S,f_S)$ such that for some $p \in P_M$, $m(p) = m(\phi(p)) = 1$ and $m(q) = 0$ for all $q \notin \{p,\phi(p)\}$. 
Then marking $f_S$ is reachable from marking $m$.
\end{proposition}
\begin{proof} We first introduce shorthand notation for ordering two transitions.
Let $t, t'$ be two transitions. 
We write $t \sendreceive t'$ to mean $tt'$ if $\lambda(t) = send$, and $t't$ otherwise.
We next focus on proving the following property:\medskip

\noindent {\em $\forall k \in \mathbb{N}$ and $\forall m \in \mathcal{R}(S,i_S)$ such that for some $p \in P_M$, $m(p) = m(\phi(p)) = 1$ and $m(q) = 0$ for all $q \notin \{p,\phi(p)\}$, 
if $t_1\dots t_k$ is the shortest firing sequence to move the token from $p$ to $\textit{fin}_M$ in the skeleton of $M$ is of length $k$, 
then firing $(t_1 \sendreceive \phi({t}_1)) \dots (t_k \sendreceive \phi({t}_k))$ in $S$ from marking $m$ moves the tokens in $p$ and $\phi(p)$ to $\textit{fin}_M$ and $\textit{fin}_N$, which is the marking $f_S$.}\medskip

\noindent We use induction on $k$ to prove this property.
\begin{itemize}
    \item Base case.
    Suppose a reachable marking $m \in \mathcal{R}(S,i_S)$ is such that for some $p \in P_M$, $m(p) = m(\phi(p)) = 1$ and $m(q) = 0$ for all $q \notin \{p,\phi(p)\}$.
    Suppose that $t_1 \dots t_k$, the shortest firing sequence to move the token from place $p$ to $\textit{fin}_M$ in the skeleton of $M$, is of length $k = 0$.
    From this, it follows that $p = \textit{fin}_M$ and $\phi(p) = \textit{fin}_N$.
    Since there are no other tokens in $S$, we find that $m = f_S$, so our statement holds.
    
    \item Inductive case.
    As our induction hypothesis, we assume that for all reachable markings $m' \in \mathcal{R}(S,i_S)$ for which for some $p' \in P_M$, $m'(p') = m'(\phi(p')) = 1$ and $m'(q') = 0$ for all $q' \notin \{p',\phi(p')\}$, if  $t_1 \dots t_k$, the shortest firing sequence to move the token from place $p'$ to $\textit{fin}_M$ in the skeleton of $M$, is of length $k$, then firing $(t_1 \sendreceive \phi({t_1})) \dots (t_k \sendreceive \phi({t_k}))$ in $S$ moves the tokens in $p'$ and $\phi(p')$ to $\textit{fin}_N$ and $\textit{fin}_M$.
    
    Next, assume that $m \in \mathcal{R}(S,i_S)$ is a reachable marking for which for some $p \in P_M$, $m(p)= m(\phi(p)) = 1$ and $m(q) = 0$ for all $q \notin \{p,\phi(p)\}$.
    Suppose that $u_1 \dots u_{k+1}$, the shortest firing sequence to move the token from place $p$ to $\textit{fin}_M$ in the skeleton of $M$, is of length $k+1$.\\
    
    We distinguish two cases, based on the direction of $u_1$:
    \begin{itemize}
        \item $\lambda(u_1) = send$. 
        First observe that $u_1$ is enabled in $m$, too.
        Firing $u_1$ in $S$ produces a token in some unique place $p' \in P_M \cap \post{u_1}$.
        Since $\lambda(u_1) = send$, after firing $u_1$, there is also a token in the unique interface place $i \in \post{u_1}\setminus P_M$.
        We thus reach a marking $m'$ such that $m'(p') = m'(\phi(p)) = m'(i) = 1$ and $m'(q) = 0$ for all $q \notin \{p',\phi(p),i\}$. 
        
        Since $\phi(p)$ is mirror of $p$, we find that $\exists \bar{u}_1 \in T_N: \bar{u}_1 = \phi(u_1)$ is enabled, since $\{\phi(p),i\} = \pre{\bar{u}_1}$.
        Firing $\bar{u}_1$ consumes the tokens in $\phi(p)$ and $i$, producing a token in $\phi(p')$. 
        We therefore reach a marking $m''$ such that $m''(p') = m''(\phi(p')) = 1$ and $m''(q) = 0$ for all $q \notin \{p',\phi(p')\}$.
        Since the shortest firing sequence $u_2 \dots u_{k+1}$ from $p'$ to $\textit{fin}_N$ is of length $k$, we find, by induction, that firing $(u_2 \sendreceive \phi({u}_2)) \dots (u_{k+1} \sendreceive \phi({u}_{k+1}))$ in $S$ moves the token in $p'$ and $\phi(p')$ to $\textit{fin}_M$ and $\textit{fin}_N$.
        But then firing $(u_1 \sendreceive \phi({u}_1)) \dots (u_{k+1} \sendreceive \phi({u}_{k+1}))$ moves the token in $p$ and $\phi(p)$ to $\textit{fin}_M$ and $\textit{fin}_N$, i.e., the final marking.
        \item The argument for case $\lambda(u_1) = receive$ is fully dual.
    \end{itemize}
\end{itemize}
Consequently, we conclude our lemma follows from the statement we proved.
\qed
\end{proof}
We are now in a position to formally prove Theorem~\ref{thm:wt_portnets}:
\begin{proof}[Theorem~\ref{thm:wt_portnets}]
Let $N$, $M$ and $S$ be as stated.
Let $m \in \mathcal{R}(S, i_S, f_S)$.

By Proposition~\ref{lemma:synchronisable}, there must be a sequence of transitions $t_1,\dots,t_n$ such that $m \xrightarrow{t_1}\dots\xrightarrow{t_n} m'$ such that $\lambda(t_i) = receive$ for all $i \in 1 \ldots n$, and there is a place $p \in P_M$ such that $m'(p) = m'(\phi(p)) = 1$ and all other places of $S$ are empty.
By Proposition~\ref{lemma:homing_sequence} it follows that $f_S$ is reachable. 
Hence, $S$ is weakly terminating.
\qed
\end{proof}
The previous result is strengthened by showing that if the final places of a composition of well-formed labeled portnet and its well-formed partial mirror are marked, then no other places of the composition are marked.

\begin{theorem}[\textbf{Proper Completion}]\label{thm:proper_completion_portnets}
Let $N$ be a well-formed labeled portnet and let $(M,\phi)$ be a partial mirror of $N$. 
Let $S = \textit{compose}(\{N,M\})$ and $i_S = i_N + i_M$ be the initial marking of $S$ and $f_S = f_N + f_M$ be the final marking of $S$. 
For all markings reachable in the net system $(S, i_S, f_S)$, if the final places of $N$ and $M$ are marked then all other places are empty. 
\end{theorem}
\begin{proof}
Suppose that the final places of $N$ and $M$ are marked and there are left over tokens in other places of net $S$.  
There cannot be any left over tokens in the internal places of $N$ and $M$, since they are state machine nets having a single token (initial place) each in their initial marking. 
If there are left over tokens in interface places, then by the Thm.~\ref{thm:wt_portnets}, there must exist a firing sequence that leads to the final marking $f_S$. This is not possible since the final places of $N$ and $M$ have empty postsets, which is a contradiction. 
\qed
\end{proof}

\section{Weak Termination of Composition of Multiple Portnets}
In practice, multiple clients may wish to consume the same service provided by a server, repeatedly. 
The results we established in the previous section do not generalize to such a setting unchanged. 
To see this, consider the well-formed labeled portnets $X, Y$ and $Z$ depicted in Fig.~\ref{fig:deadlockWithMultipleClients}. 
Observe that $(Y,\phi_1)$ and $(Z,\phi_2)$ are partial mirrors of $X$.
While the compositions $\textit{compose}(\{X,Y\})$ and $\textit{compose}(\{X,Z\})$ are deadlock free, the firing sequence $\langle 8, 10, 16, 2, 4, 6, 18, 19 \rangle$ in the composition $\textit{compose}(\{X, Y, Z\})$ leads to a deadlock. 
The main reason is that clients may interfere with each other. 
It turns out that in order to achieve the desired result we must embed our results in a synchronization pattern. 

\begin{figure}[t]
    \centering
    \includegraphics[scale=0.59]{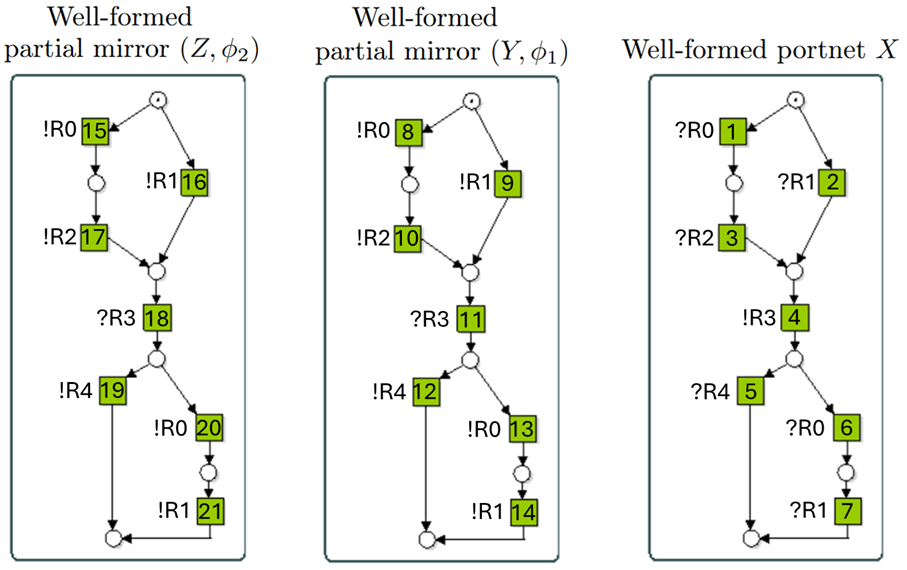}
    \caption{Deadlock in composition: $\textit{compose}(\{X,Y,Z\})$}
    \label{fig:deadlockWithMultipleClients}
\end{figure}
A synchronization pattern 
requires that a server first makes a selection among the clients with which to cooperate. 
Such negotiation can be modeled by means of the pattern shown in Fig.~\ref{fig:synchPattern}. 
Starting from the initial marking, one or more clients may send their requests, one of which is picked up by the server, followed by a response back. 
The response is picked up by exactly one client, which guarantees that exactly one client may progress with the server to eventually mark their respective final places without interference from others. 
Once the final place of the server is marked a closure transition $t$ returns the token back to the initial place, and from thereon the server can pick the next client.
This process repeats until all clients have been handled and the final marking is reached. 
\begin{figure}[htb]
    \centering
    \includegraphics[scale=0.37]{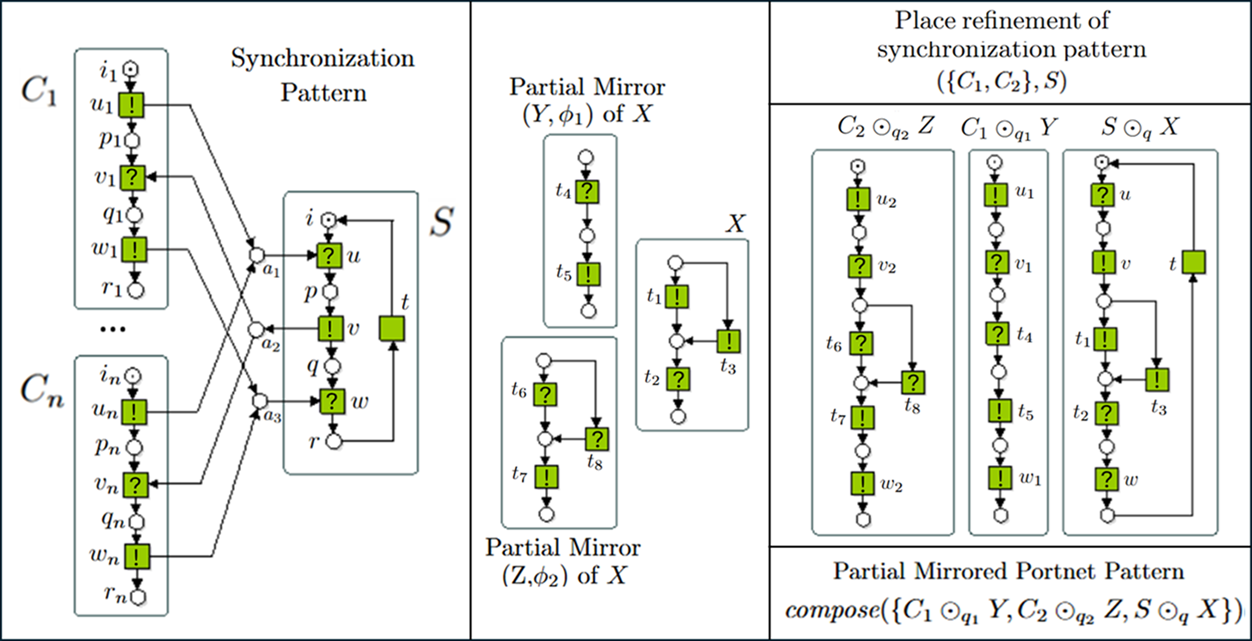}
    \caption{Synchronization Pattern and Partial Mirrored Portnet Pattern}
    \label{fig:synchPattern}
\end{figure}
\begin{definition}[Synchronization Pattern]\label{def:synchronization_pattern}
The pattern depicted in Fig.~\ref{fig:synchPattern} (left) is a \emph{synchronization pattern} $(\mathcal{C}, S, Q)$ 
where $S$ is a closure of a well-formed labeled portnet $N$, $\mathcal{C} = \{C_1, C_2,  \ldots C_n\}$ is a set of $n$ ($n \in \mathbb{N}$) well-formed mirrored clients of $N$, and $Q$ is the set of refinable places $\{q, q_1, q_2, \ldots q_n\}$.
The net system is defined as $(\hbox{compose}(\mathcal{C} \cup \{S\}), m_0, m_f)$, where $m_0$ is the initial marking given by $\{i\} \cup \{ i_j \mid j \leq n\}$, 
and final marking $m_f$ given by 
$\{i\} \cup \{r_j \mid j \leq n\}$. 
\end{definition}
\begin{lemma}
\label{thm:weakly_terminating_sync_pattern}
The synchronization pattern is weakly terminating.
\end{lemma}
\begin{proof}
By the Thm. 3.3.10 \cite{bera2011component}, the composition of multiple client portnets with a server portnet is weakly terminating. The synchronization pattern is such a composition, therefore it is also weakly terminating. 
\qed
\end{proof}
The structure of a synchronization pattern leads to the following observations: First, since the skeletons of client and server are S-nets with a single token in \textit{init}, there can be at most one token in their internal places.  
Second, exactly one token produced by one of the clients can be consumed by the server. The server can then produce exactly one token back towards its clients. Third, at most one client consumes a token produced by the server and it becomes the chosen client. when a client is chosen the mirror places of that client and server are marked (first server then client) and all interface tokens produced by one were consumed by the other. From here the chosen client produces a token and marks its final place, enabling the server to pick this token and mark its final place. 
\begin{lemma}\label{placeInvariants} 
Let $m$ be a marking of a synchronization pattern with $n \in \mathbb{N}$ clients satisfying the following properties
\begin{enumerate}
    \item $m(i) + m(p) + m(q) + m(r) = 1$
    \item $\forall j \in \{1 \ldots n\}: m(i_j) + m(p_j) + m(q_j) + m(r_j) = 1$
    \item $\sum_{j \in \{1 \ldots n\}} m(p_j) = m(a_1) + m(p) +m(a_2)$,
    \item $ (\sum_{k \in \{1 \ldots n\}} m(q_k)) + m(a_2)\leq 1$, 
    \item $m(a_2) + m(a_3) \leq 1$ and $m(a_2) + m(a_3) = m(q) - \sum_{k \in \{1 \ldots n\}} m(q_k)$
\end{enumerate}
Then firing any enabled transition $t$ from $m$ does not violate these properties. 
\end{lemma}
\begin{proof}
Suppose $m$ is a marking satisfying the stated properties.

\begin{itemize}
    \item Suppose transition $u$ is enabled. Then $m(i) > 0 \wedge m(a_1) > 0$. The firing of transition $u$, consumes a token from $i$ and $a_1$ and adds a token to $p$. This changes the number of tokens in the places of invariants $1$ and $3$ as follows
    $$m(i) - 1 + m(p) + 1 + m(q) + m(r) = 1$$ and $$\sum_{j \in \{1, \ldots, n\}} m(p_j) = m (a_1) - 1 + m(p) + 1 + m(a_2)$$ 
    Hence firing $u$ leads to a marking where these invariants are again satisfied.\\ 
    
    \item Suppose transition $v$ is enabled. Then $m(p) > 0$. By the invariant $1$, we deduce that $m(q) = 0$. So in the invariant $5$, $m(q) - \sum_{k \in \{1, \ldots, n\}} m(q_k) = 0$, which means $m(a_2) + m(a_3) = 0$. 
    
    The firing of transition $v$, consumes a token from $p$ and adds a token to $q$ and $a_2$. This changes the number of tokens in the places of invariants $1$, $3$, $4$ and $5$ as follows 
    $$m(i) + m(p) - 1 + m(q) + 1 + m(r) = 1$$
    $$\sum_{j \in \{1, \ldots, n\}} m(p_j) = m(a_1) + m(p) - 1 + m(a_2) + 1$$
    $$ \sum_{k \in \{1, \ldots, n\}} m(q_k) + m(a_2) + 1 \leq 1$$
    $$m(a_2) + 1 + m(a_3) \leq 1$$
    $$m(a_2) + 1 + m(a_3) = m(q) + 1 - \sum_{k \in \{1 \ldots n\}} m(q_k)$$
    Hence firing $u$ leads to a marking where these invariants are again satisfied.\\ 
    
    \item Suppose transition $w$ is enabled. Then $m(q) > 0$ and $m(a_3) > 0$. By the invariant $5$, we can conclude that $m(a_2) = 0$. 
    The firing of transition $w$, consumes a token from $q$ and $a_3$ and adds a token to $r$. 
    This changes the number of tokens in the places of the invariants $1$ and $5$ as follows 
    $$m(i) + m(p) + m(q) - 1 + m(r) + 1 = 1$$
    $$m(a_2) + m(a_3) - 1 \leq 1$$
    $$m(a_2) + m(a_3) - 1 = m(q) - 1 - \sum_{k \in \{1, \ldots, n\}} m(q_k)$$
    Hence firing $w$ leads to a marking where these invariants are again satisfied.\\ 
    
    \item Suppose transition $t$ is enabled. Then $m(r) > 0$. The firing of $t$, consumes a token from $r$ and adds a token to $i$. This changes the number of tokens in the places of invariant $1$ as follows $$m(i) + 1 + m(p) + m(q) + m(r) - 1 = 1$$ 
    Hence firing $t$ leads to a marking where these invariants are again satisfied.\\
    
    \item Suppose transition $u_j$ is enabled for some $j \in \{1, \ldots, n\}$. Then $m(i_j) > 0$. The firing of $u_j$ consumes a token from $i_j$ and adds a token to $a_1$ and $p_j$. This changes the number of tokens in the places of invariant $2$ and $3$ as follows
    $$m(i_j) - 1 + m(p_j) + 1 + m(q_j) + m(r_j) = 1$$
    $$(\sum_{j \in \{1, \ldots, n\}} m(p_j)) + 1 = m(a_1) + 1 + m(p) +m(a_2)$$
    Hence firing $u_j$ leads to a marking where these invariants are again satisfied.\\

    \item Suppose transition $v_j$ is enabled for some $j \in \{1 \ldots n\}$. Then $m(a_1) > 0 \wedge m(p_j) > 0$. The firing of transition $v_j$ consumes a token from $p_j$ and $a_2$ and adds a token to $q_j$. This changes the number of places in the invariants of $2$, $3$, $4$ and $5$ as follows
    $$m(i_j) + m(p_j) -1 + m(q_j) + 1 + m(r_j) = 1$$
    $$(\sum_{j \in \{ 1, \ldots, n\}} m(p_j)) - 1 = m(a_1) + m(p) + m(a_2) - 1$$
    $$(\sum_{j \in \{1, \ldots, n\}} m(q_j)) + 1 + m(a_2) - 1 \leq 1$$
    $$m(a_2) - 1 + m(a_3) = m(q) - (\sum_{j \in \{1, \ldots, n\}} m(q_j) + 1)$$
    
    Hence firing $v_j$ leads to a marking where these invariants are again satisfied.\\
    
    \item Suppose transition $w_j$ is enabled for some $j \in \{1, \ldots, n\}$. Then $m(q_j) > 0$. 
    
    By the invariant $4$, we can deduce that $m(a_2) = 0$. Furthermore, by the invariant $1$ and $5$, we can deduce that $m(q) = 1$.     
    The firing of transition $w_j$ consumes a token from $q_j$ and adds a token to $r_j$ and $a_3$. This changes the number of tokens in the places of invariants $2$, $4$ and $5$ as follows
    $$m(i_j) + m(p_j) + m(q_j) - 1 + m(r_j) + 1 = 1$$
    $$(\sum_{k \in \{1, \ldots, n\}} m(q_k) - 1) + m(a_2) \leq 1$$
    $$m(a_2) + m(a_3) + 1 \leq 1$$
    $$m(a_2) + m(a_3) + 1 = m(q) - (\sum_{k \in \{1, \ldots, n\}} m(q_k) - 1)$$
    Hence firing $w_j$ leads to a marking where these invariants are again satisfied.
\end{itemize}
We thus conclude that firing any enabled transition from a marking that satisfies all 5~properties again results in a marking that satifies these properties. \qed
\end{proof}
\begin{corollary}\label{cor:sync_places_sync_pattern}
The net system of a synchronization pattern 
with $n \in \mathbb{N}$ clients, satisfies 
\begin{itemize}
    \item $\forall m \in \mathcal{R}(N,m_0) : m(q) = 1 \implies \sum_{j \in \{1 \ldots n\}} m(q_j) \leq 1$
    \item $\forall m \in \mathcal{R}(N,m_0) : \sum_{j \in \{1 \ldots n\}} m(q_j) = 1 \implies m(q) = 1$
\end{itemize}
\end{corollary}
\begin{proof}
Follows from Lemma~\ref{placeInvariants} and the fact that marking $m_0$ satisfies the place invariants (Def.~\ref{def:synchronization_pattern}).
\qed
\end{proof}
Refinements of places and transitions \cite{murata1989petri} are a well-known strategy to incrementally elaborate the behavior of a Petri net while preserving soundness. We extend the classical place refinement operation to labeled OPNs.  
\begin{definition}[Place Refinement of Labeled OPN]\label{def:place_refinement}
Given two disjoint labeled OPNs $N$ and $M$. 
The refinement of a place $p \in P_N$ by $M$, denoted by $N \odot_p M$, is a labeled OPN $(P, I, O, T, F, \textit{init}, \textit{fin}, \mu)$ with
$P = (P_N \setminus \{p\}) \cup P_M$, 
$I = I_N \cup I_M$,
$O = O_N \cup O_M$,
$T = T_N \cup T_M$,  
$F = (F_N \setminus ((\pre{p} \times \{p\}) \cup (\{p\} \times \post{p}))) 
\cup F_M \cup (\pre{p} \times \textit{init}_M) \cup (\textit{fin}_M \times \post{p})$,
$\textit{init} = \textit{init}_N$, $\textit{fin} = \textit{fin}_N$, 
$\mu = \mu_N \cup \mu_M$.
\end{definition}
We formally introduce the notion of a \emph{partial mirrored portnet pattern}. Such a pattern captures the essence of a server communicating with multiple clients, obtained by (partial) mirroring, embedded inside a synchronization pattern by means of place refinements. An example is shown in Fig.~\ref{fig:synchPattern}.
\begin{definition}[\textbf{Partial Mirrored Portnet Pattern}]\label{def:partial_mirrored_portnet_pattern}
Let $N$ be a weakly terminating well-formed labeled portnet and let $(M_1,\phi_1), \ldots, (M_n,\phi_n)$, for $n \in \mathbb{N}$, be a set of disjoint partial mirrors of $N$.
Let $R = (\mathcal{C}, S, \{q, q_1, \ldots q_n\})$ be a synchronization pattern with $n$ clients such that $M_1, \ldots, M_n, N, \mathcal{C}, S$ have disjoint skeletons. 
A \emph{partial mirrored portnet pattern} is a composition $L = \hbox{compose}(\mathcal{C}' \cup \{\mathcal{S}'\})$, where 
$\mathcal{C}' = \{ C_1 \odot_{q_1} M_1, \ldots, C_n \odot_{q_n} M_n \}$ and $\mathcal{S}' = S \odot_{q} N$
with initial marking $m_0 = i_S + i_{C_1} + \ldots + i_{C_n}$, final marking $m_f = i_S + f_{C_1} + \ldots + f_{C_n}$ and a net system $(L,m_0,m_f)$. If $M_1, \ldots, M_n$ are full mirrors of $N$, then we call $L$ a \emph{mirrored portnet pattern}. 
\end{definition}

\begin{theorem}\label{thm:partial_mirrored_portnet_weakly_terminates}
The partial mirrored portnet pattern weakly terminates.
\end{theorem}
\begin{proof}
By Lemma.~\ref{thm:weakly_terminating_sync_pattern} we know that a synchronization pattern consisting of $n \in \mathbb{N}$ clients weakly terminates. 
By Cor.~\ref{cor:sync_places_sync_pattern}. we know that first a server needs to mark a place $q$ and only then exactly one client $C_j$, for some $j \in \{1 \ldots n\}$, can mark its place $q_j$. Furthermore, the server cannot remove a token from $q$ until $q_j$ is marked. In the skeleton of the server, the place $q$ is always on a path from the initial place $i$ to final place $r$. So all firing sequences from the initial marking to the final marking of a synchronized portnet, first marks place $q$ with a token, followed by place $q_j$, while $q$ remains marked.

By the Def.~\ref{def:place_refinement} and Def.~\ref{def:partial_mirrored_portnet_pattern}, 
only the initial and final place of 
a well-formed labeled portnet $N$ and its partial mirrors $(M_1,\phi_1),\ldots,(M_n,\phi_n)$
have nodes (transitions) of the synchronization pattern in their preset and postset, respectively. 
So a marking of a synchronization pattern $R$ containing one token in $q$ and $q_j$, corresponds to a marking in the partial mirrored portnet pattern $L$ with a token in initial places of $N$ and $M_j$, respectively. 
By Thm.~\ref{thm:wt_portnets}, we know that from all reachable markings from such a marking, there exist a firing sequence to the final marking, i.e., final places of $N$ and $M_j$. Furthermore, by Thm.~\ref{thm:proper_completion_portnets}, we know that when the final places of $N$ and $M_j$ are marked, then all other places of $N$ and $M_j$ are empty.
By the Def.~\ref{def:place_refinement} and Def.~\ref{def:partial_mirrored_portnet_pattern},
a marking of $L$ containing a token in the final places of $N$ and $M_j$ (and no other tokens in places of $S$), enables exactly the same set of transitions as that of the net $R$ with a token in the places $q$ and $q_j$. 
Hence, the refinement with a multi-client portnet does not violate the weak termination property of the synchronization pattern. 
\qed
\end{proof}

\paragraph{Relation to Component Framework.} 
The component-based architectural framework introduced in \cite{BeraThesis} is based on the two concepts: components and interaction patterns. 
A \emph{component} is a strongly connected OPN with a special place called idle, which means that components are cyclic in their execution and therefore handle multiple clients. 
An \emph{interaction pattern} is a class of parameterized multi-workflow nets defined as a composition of a set of identical clients (OWN) and a set of identical servers (components). Four recurring weakly terminating interaction patterns are defined, namely remote procedure call (RPC), publish-subscribe (PS), goal-feedback-result (GFR) and mirrored port (MP). 
To represent a system design as a composition of components over interaction patterns, a graph based architectural diagram notation is presented. 
To derive a Petri net from an architectural diagram, a construction method is prescribed and shown to guarantee the weak termination property. 
The construction method is a refinement strategy based on \emph{simultaneous refinement} and \emph{insertion} operations to incrementally introduce interaction patterns into a system (OPN). 

To add a new pattern to this framework, it suffices to show that it is a weakly terminating interaction pattern, since~\cite[Lemma~3.4.1]{BeraThesis} guarantees that weak termination is preserved under simultaneous refinement. 
By Thm.~\ref{thm:partial_mirrored_portnet_weakly_terminates} we know that a mirrored portnet pattern weakly terminates. A synchronization pattern is an instance of a MP pattern, which is an interaction pattern. 
A mirrored portnet pattern is obtained by place refinements of refinable places of a synchronization pattern with S-OWNs, corresponding a well-formed server and its mirrored clients, followed by their composition, which is once again an interaction pattern. 

\begin{proposition}[\textbf{Refinement of Synchronization Pattern}]\label{lemma:sync_pattern_is_interaction_pattern}
A mirrored portnet pattern is an interaction pattern.
\end{proposition}
\begin{proof}
By the Def.~\ref{def:synchronization_pattern}, a synchronization pattern is an interaction pattern \cite{BeraThesis} with empty inhibitor and reset arcs. 
By the Def.\ref{def:composability}, a composition of a well-formed server portnet and its full mirrored clients satisfies (i) for any client and server pairs, input places of a client is equal to the output place of a server and vice-versa, (ii) transitions of client and server connected to the same interface place have opposite directions. 
By the Def.~\ref{def:partial_mirrored_portnet_pattern}, the successive place refinements (Def.~\ref{def:place_refinement}) of refinable places of a synchronization pattern with a well-formed server portnet and its full mirrored clients results once again in server that is a closure of a S-OWN and a set clients that are S-OWNs. Their composition is an interaction pattern \cite{BeraThesis} with empty inhibitor and reset arcs. 
\qed
\end{proof}

\section{Application}
The growing recognition that a lack of explicit and precise specifications of component interfaces potentially leads to problems during system integration have led to the development of an open-source interface modeling tool ComMA \cite{kurtev2024model}. The goal is to capture both static and dynamic aspects of an interface as a sound basis for a formal contract between service providers and consumers. 
ComMA specifications are used to validate whether implementations conform to their specifications using run-time monitoring \cite{kurtev2017runtime} and more recently, model-based testing techniques. 
The tool is developed in collaboration with industry and is used by software engineers at Philips and Thales as part of their development processes. To fully realize the benefits and detect design errors earlier, the need for correct interface specifications was recognized. 
The results of this paper have been implemented in ComMA as algorithms to check well-formedness of interface specifications.
The feature has proven to be successful in raising awareness about the pitfalls of designing asynchronous systems and serve as a guide for engineers to create weakly terminating interface protocols by construction. 
\begin{figure}[t] 
    \centering
    \includegraphics[scale=0.82]{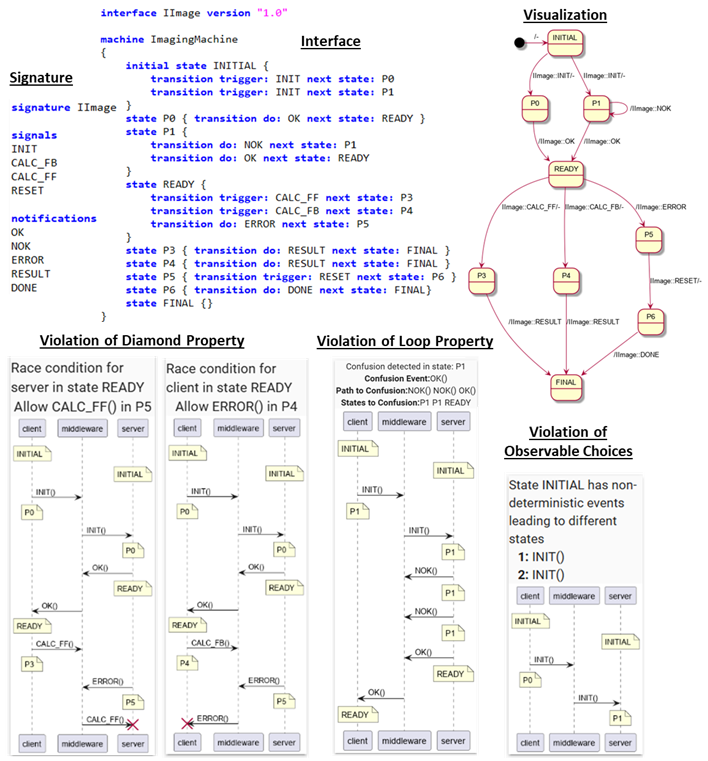}
    \caption{Application of the results to ComMA}
    \label{fig:ComMA}
\end{figure}
ComMA provides a family of domain specific languages that includes types, signature, interface and input parameters. \emph{Types} capture data definitions built of basic types (int, real, string, etc.), enumerations, vectors, records and maps. 
\emph{Signature} captures event definitions (i.e., their labels and data parameter types) of type: command-reply (client sends a command and is blocked until server reply), signal (client sends a non-blocking event) and notification (server sends a non-blocking event). 
\emph{Interface} captures allowed order of input and output events as transitions of a protocol state machine, defined from the perspective of a service provider (server), also known as a \emph{provided interface}. 
The state machine may additionally define global variables and manipulate them using an expressions language. Output events containing data arguments refer to global variables, while data arguments of input events are provided by clients. 
When defining usage of a provided interface by another component, it is common practice to simply mirror the protocol of the server. So an explicit model of a client protocol is omitted. 
For standalone simulation of the server state machine, \emph{input parameters} specify possible data values of input event arguments (in each possible state). 
Analysis and simulation of server state machine is made possible by transformation to colored Petri nets using SNAKES \cite{pommereau2015snakes}. A formal treatment of this transformation is out of scope. 
The remainder of this section will omit data aspects and informally discuss the mapping to labeled portnets.

A ComMA specification of the server portnet $S$ in Fig.~\ref{fig:deadlocked} is shown in Fig.~\ref{fig:ComMA}. The signature defines signals and notifications with names corresponding transition labels of $S$. The syntax of an interface state machine comprises of 
state labels (with an explicit initial state) and transitions are directed edges between states. A generated visualization of the state machine in UML syntax is shown in Fig.~\ref{fig:ComMA}. 
Each state maps to a place in the underlying portnet with a name corresponding its label. 
Each state defines a set of possible transitions and each transition defines a target \emph{next state}. 
Two kinds of transitions are distinguished. 
A \emph{non-triggered transition} defines an initiative of a server in a 
\emph{do} section containing a 
notification which maps to a \emph{send} transition with a preset corresponding the state (place) it is defined within and a postset corresponding the state (place) referred to by \emph{next state}. 
A \emph{triggered transition} defines an initiative of a client in a 
\emph{trigger} section containing a signal which maps to a \emph{receive} transition with a preset corresponding the state (place) it is defined within and a postset corresponding the state (place) referred to by \emph{next state}. 
With these mapping rules, a labeled portnet is derived from a state machine syntax.  
The underlying translation discards interface places since we are only interested in analyzing the skeleton. 
For compact specifications, many shorthand notations are available, e.g., triggered transitions with multiple \emph{do} sections, each having a sequence of notifications, \emph{for all states} prevents duplication of same transition in multiple states, etc. 
All constructs can be expressed in terms of concepts presented so far. 

To support a modeler in creating weakly terminating specifications, algorithms are implemented in ComMA to first validate whether the skeleton is weakly terminating and then to validate whether the observable choices, diamond and loop properties hold. 
Whenever a violation is detected, a UML sequence diagram is generated (see Fig.~\ref{fig:ComMA}) illustrating the problem in a composition of a server with an assumed client that fully mirrors the protocol. 

\section{Conclusion}

This paper presented sufficient conditions to guarantee weak termination in a composition of a server and multiple \emph{partial} mirrored client protocols. 
These conditions impose structural restrictions characterized by well-formedness, relaxing the strict restrictions imposed by existing results, thereby encompassing a larger class of communication protocols encountered in practice.
By embedding well-formedness checks on the server protocol and assuming a symmetric mirrored client in the ComMA tool, modelers are made aware about potential hidden issues in component interactions, early on during design-time. 
As future work, we aim to generalize the synchronization pattern, introduced to preserve weak termination in a composition of multiple clients. 
We also aim to extend our results to the component modeling language of ComMA, where functional constraints define restrictions over events of provided and required interfaces. 

\section*{Acknowledgements}
The research is carried out as part of the \emph{TechFlex 2026} and \emph{Matala 2025} program under the responsibility of TNO-ESI in cooperation with ASML Holding N.V. and Thales nederland B.V. The research activities are co-funded by Holland High Tech | TKI HTSM via the PPP Innovation Scheme (PPP-I) for public-private partnerships. 

\bibliographystyle{splncs04}
\bibliography{libary}

\end{document}